\documentclass{tMAM2e}
\usepackage{amsmath,amsthm}
\usepackage{graphicx,epsfig,caption,minibox,multirow,makecell,tikz-cd,xcolor}
\usepackage{natbib}

\usepackage[shortlabels]{enumitem}
\usepackage{pgfplots}
\usepackage{rotating}
\usepackage{relsize}
\tikzset{fontscale/.style = {font=\relsize{#1}}
    }

\newcommand\Tstrut{\rule{0pt}{3.0ex}}         
\newcommand\Bstrut{\rule[-1.5ex]{0pt}{0pt}}   

\usepackage{mathtools}
\usepackage{pdfpages}
\usepackage{blkarray}
\usetikzlibrary{decorations.markings}
\usetikzlibrary{calc}

\usepackage[utf8]{inputenc}

\usepackage{graphics}
\usepackage{tikz}
\usetikzlibrary{arrows}
\usepackage{listings}

\usepackage{amssymb}

\usepackage{cancel}

\theoremstyle{plain}
\newtheorem{theorem}{Theorem}[section]

\newtheorem{proposition}[theorem]{Proposition}

\theoremstyle{definition}
\newtheorem{definition}[theorem]{Definition}

\theoremstyle{remark}

\newtheorem{example}[theorem]{Example}

\begin{document}

\title{Persistent Homology of Music Network with Three Different Distances}
\author{
\name{Eunwoo Heo\textsuperscript{ab}, Byeongchan Choi\textsuperscript{ab}, Myung ock Kim\textsuperscript{c}, Mai Lan Tran \textsuperscript{ab},  Jae-Hun Jung \textsuperscript{ab}$^{\ast}$\thanks{$^\ast$Corresponding author. Emails: hew0920@postech.ac.kr (Eunwoo Heo), bcchoi@postech.ac.kr (Byeongchan Choi),  kmyo332@gmail.com (Myung ock Kim), mailantran@postech.ac.kr (Mai Lan Tran),  jung153@postech.ac.kr (Jae-Hun Jung)}}
\affil{\textsuperscript{a} Department of Mathematics, Pohang University of Science \& Technology, Pohang 37673, Korea; 
\textsuperscript{b} POSTECH Mathematical Institute for Data Science (MINDS), Pohang University of Science \& Technology, Pohang 37673, Korea;
\textsuperscript{c}Korea Institute for Advanced Study (KIAS), Seoul 02455, Korea}
\received{v6.0 released January 2024}
}

\maketitle

\begin{abstract}
Persistent homology has been widely used to discover hidden topological structures in data across various applications, including music data. To apply persistent homology, a distance or metric must be defined between points in a point cloud or between nodes in a graph network. These definitions are not unique and depend on the specific objectives of a given problem. In other words, selecting different metric definitions allows for multiple topological inferences. In this work, we focus on applying persistent homology to music graph with predefined weights. We examine three distinct distance definitions based on edge-wise pathways and demonstrate how these definitions affect persistent barcodes, persistence diagrams, and birth/death edges. We found that there exist inclusion relations in one-dimensional persistent homology reflected on persistence barcode and diagram among these three distance definitions. We verified these findings using real music data. 
\end{abstract}

\begin{keywords}
Topological data analysis; persistent homology;  music graph network; metric; one-dimensional homology; inclusion relation
\end{keywords}

\begin{classcode}\textit{Classification codes}: {\color{black}{AMS 00A65, 55N31}} \end{classcode}

\section{Introduction}
Topological Data Analysis (TDA), developed relatively  recently compared to other traditional data analysis tools, has proven its usefulness in various applications including music analysis~\citep{Carlsson,Cohen,Eharer,ZC2005}. Unlike traditional statistical approaches, TDA focuses on the shape of the data, providing insights into its geometry and topology. By using topological information, novel interpretations about the data can be found. Persistent homology (PH) is one of the primary tools in TDA, utilizing homological characteristics of the data across various scales~\citep{Carlsson}. 
Essentially, when representing the data's shape as a complex, the PH method computes the homology of the corresponding complex. 
This approach involves constructing the complex in different scales sequentially, capturing how homology changes as the complex is built and refined. This procedure is known as the {\it filtration}. This hierarchical understanding of homology changes enables the PH method to infer the topological properties of the given data's shape. Particularly, PH shows the cyclic structures within the given data, making it a useful concept in music analysis.
 
For the PH method, it is crucial to transform the given data into a suitable metric space. Typically, most data lack an intrinsic geometric structure and need to be transformed into an appropriate geometric object. 
Graph representation is widely used in network science for PH analysis~\citep{Aktas2019}, including music data. 
In~\citep{Bergomithesis,Bigo,BWold,Tran02012024,Tran02092023}, PH analysis has been applied to music data, with studies using the concept of nodes and connectivity (edges) of each music note in the music flow, constructing the music network as a graph. For example, a tuple comprising the pitch and duration of the music note defines a node, and an edge is formed when these two nodes are placed side by side in a music flow. 

Once the appropriate transform is identified, applying the PH method requires defining a proper metric for the transformed geometric object, determining the distance between points in the point cloud or between nodes in the graph. There is no single way to define the metric, as there can be multiple definitions. 
The simplest approach is to adopt the Euclidean metric, but various other definitions exist, making the definition non-unique. 
Consequently, some definitions yield superior performance and provide richer information about the data, helping in better understanding the data. Moreover, the distance used for PH analysis may not always satisfy the necessary conditions of a metric, particularly the triangle inequality. For instance, the distance definition used in music analysis~\citep{Tran02012024,Tran02092023} is reasonable and consistent with our intuition, but violates the triangle inequality, which will be shown in this paper. But this violation is not crucial for practical data analysis. Despite not being a metric, the distance definition still yields meaningful insights into music data.

Given the non-uniqueness of metric or distance definitions for the constructed point cloud or graph, it becomes an important task to: 1) determine whether the introduced distance satisfies the metric definition, and 2) observe how the PH structure varies as the distance definition changes for music analysis. Despite the abundance of PH research and applications, studies focusing on these tasks are limited. In~\citep{Sethares02012014}, an analysis was conducted on how PH changes with different metrics on musical data, such as necklace distance, pitch-class distance, chord-class distance, and rhythm distance. It was shown that PH analysis based on these distances shows well-known topological structures in music. These definitions satisfy the metric conditions.

In this paper, our focus is on music graph data, investigating various tasks. While our initial motivation stemmed from the graph derived from music data, this study can be extended to any graph. 
In this paper, we introduce three distinct distance definitions for music data based on pathway distances between nodes on the graph network. These definitions are commonly used when combined with the PH method. The first distance between the $i$th and $j$th nodes, $n_i$ and $n_j$, is defined as the path distance involving the smallest number of intermediate nodes connecting $n_i$ and $n_j$. The second distance is defined as the path distance that results in the shortest overall distance between consecutive intermediate nodes. The last distance is defined as the distance that falls between the first and second distance values. These definitions could be purely geometric and not musical, as they simply focus on the connectivity between nodes within the graph network. We then show how these definitions satisfy metric definitions. As a main result, we illustrate how the PH structure of music data changes with different distance definitions. Furthermore, we show that there exists an inclusion relation among them and validate this claim through real music examples. 
This result is consistent with recent research showing that such an inclusion relation is universal when the distance definitions are based on paths between nodes in a graph~\citep{Heo2024}. 
These findings provide interesting insights into understanding music structure, particularly regarding the inclusion relation. We observed that depending on the distance definition, cycle information and persistence bar structures can vary, including the number of cycles, cycle elements, birth and death edges of cycles, etc. The inclusion relation indicates that these changes are not random but systematic. 

The musical implication of such a relation must be interesting and needs to be further investigated. Another implication of these findings is their potential utility in AI music composition. In~\citep{Tran02012024}, PH analysis was employed for automatic machine composition. Particularly when the available music data for training is limited, machines can be trained with the topological structure of the given music, producing pieces that imitate this structure. 
Different distance definitions in PH analysis yield distinct PH structures, which can serve as a tool for AI music composition. 

This paper is structured as follows: Section 2 briefly summarizes persistent homology, which is the main methodology for the analysis of music data in our work.
Section 3 provides a brief introduction to music network as a graph, illustrating how to construct a graph from music data and define distances. Then it illustrates how to define nodes and edges for the graph. Associated with the constructed graph, we will consider three different definitions of distance between nodes.  We will show how these definitions satisfy metric conditions and provide several properties induced by these definitions. 
Particularly, we will show that there exist interesting inclusion relations among graphs by these three different distance definitions. Section 4 elaborates on the inclusion relation and its representations in the persistence barcode and diagram. This section will also provide real examples from real music pieces to validate the claims made in Sections 3 and 4. Specifically, nine Korean traditional music pieces and one Western music piece are employed for numerical validation, with results consistent with the claims introduced in the previous sections. Finally, Section 5 provides a brief concluding remark.
\section{Persistent homology on metric space}
In this section, we introduce persistent homology, which will be applied to music graph in the subsequent sections. Our focus is on 
utilizing persistent homology to compute the $n$th homology of a sequence of simplicial complexes constructed in a nested manner.

Let $X$ be a topological space and $P$ be a point cloud sampled out of $X$. We are interested in $P$ as a graph network and computing its homology. 
To explain persistent homology, we consider a $k$-simplex and a simplicial complex. 
Basically, we are building up a space using simplices to mimic $X$. Let $\sigma_k$, denoted by $[x_0 x_1 \cdots x_k]$, be a convex hull composed of the $k+1$ geometrically independent vertices. By definition, the $0$-simplex corresponds to a point, the $1$-simplex a line segment, and the $2$-simplex a filled triangle. The $k$-simplex is the $k$-dimensional equivalent to a triangle. The boundary of $\sigma_k$ is defined by the formal sum over the $k-1$ simplices: $\partial \sigma_k = \sum_{i=0}^k (-1)^i [x_0 x_1 \cdots x_{i-1} x_{i+1} \cdots x_k]$. That is, for the boundary we restrict to the face of $\sigma_k$. Here, the term $(-1)^i$ is added for the orientation. For example, for the 2-simplex, $\sigma_2 = [x_0 x_1 x_2]$, the boundary is given by $\partial \sigma_2 = [x_1 x_2] - [x_0 x_2] + [x_0 x_1]$. Then, it is obvious that $\partial^2 \sigma_2 = 0$. 
This relation is valid for any $k$, i.e., $\partial^2 \sigma_k = 0$, which is the fundamental theorem of algebraic topology. 
Then, we consider the chain groups of $C_k$ of $X$, which is given by the formal sum of $\sigma_k$ of $X$. For example, if we consider a space of $X$,  $X = \sigma_2 = [x_0 x_1 x_2]$, then we have $C_0 = <[x_0], [x_1], [x_2]>$, $C_1 = <[x_0 x_1], [x_1 x_2], [x_0 x_2]>$, $C_2 = <[x_0 x_1 x_2]>$, and $C_k = 0$ for $k >2$ for this example. Then, we can define the $n$th boundary map $\partial_n: C_n \rightarrow C_{n-1}$ using the boundary map for the simplex defined above. We define $\partial_0 = 0$, i.e. $\partial_0 C_0 = 0$ which is a natural definition. For example, for $X = \sigma_2 = [x_0 x_1 x_2]$, $ \partial_1 C_1 = \partial <[x_0 x_1], [x_1 x_2], [x_0 x_2]> = <[x_1] - [x_0], [x_2] - [x_1], [x_2] - [x_0]>$. Then, the $n$th homology group of $X$, $H_n(X)$  is given by the quotient group between  $Ker (\partial_n)$ and $Im (\partial_{n+1})$ where $Ker (\partial_n)$ is the kernel group or cycle group of $\partial_n$ and $Im (\partial_{n+1})$ is the image group of $\partial_{n+1}$. 
That is, 
\begin{equation*}
H_n(X) = Ker (\partial_n)/Im (\partial_{n+1}).     
\end{equation*}
For example, if $X = \sigma_2 = [x_0 x_1 x_2]$,  $Ker(\partial_0) = C_0$ since $\partial_0 C_0 = 0$ and $Im(\partial_1)$ = $<[x_1] - [x_0], [x_2] - [x_1], [x_2] - [x_0]> = <[x_1] - [x_0], [x_2] - [x_1]>$. Thus $H_0(X) \cong \mathbb{Z}^1$. 
In a similar way, $H_1(X) = 0$. The number of generators of $H_n(X)$ is known as the $n$ Betti number, denoted by $\beta_n(X)$. Roughly speaking, $\beta_0$ is the number of the connected components of $X$ and $\beta_n, n\ge 1$ is the number of the $n$-dimensional holes of $X$. For example, for $X = S^1$, $\beta_0(X) = 1$ and $\beta_1(X) = 1$ and for $X = S^1 \times S^1$, $\beta_0(X) = 1$, $\beta_1(X) = 2$, and $\beta_2(X) = 1$.  

We consider the filtered simplicial complex, $\mathcal{K}$, a nested sequence of simplicial complexes $\emptyset = \mathcal{K}_0 \subset \mathcal{K}_1 \subset \mathcal{K}_2 \cdots \subset \mathcal{K}_m = \mathcal{K}$ and $\mathcal{K}_i = \mathcal{K}, i \ge m$. The index of $\mathcal{K}$ denotes the filtration procedure of $\mathcal{K}$. 
In the following, the sequence is constructed based on the increasing values of \( \epsilon(i) \), where \( \epsilon: \mathbb{N} \to \mathbb{R}_{\geq 0} \).

Now, we explain persistent homology~\citep{Carlsson} of $X$. For persistent homology,  we consider the sampled data out of $X$ instead of $X$ itself, denoted by $P$. In general, data, as a finite set of points, is given that describes the original topological space $X$. Such a point cloud of $P$ is used to build up a space that mimics $X$. Using the given data points, we build up the simplicial complex by gluing identified simplices. 
The identification goes through the so-called filtration procedure.
The filtration procedure provides a rule for forming simplices and connecting them in the construction of the desired complexes.
For the construction, the notion of metric or distance $d: X \times X \rightarrow \mathbb{R}$ is crucial. The definition of $d$ is not unique but rather depends on the problem and $X$ considered. In this paper, we will investigate how homology changes if the distance definition of $d$ changes. Suppose that $X$ is a metric space with $d$. Given $X$ and $d$, we construct the filtered simplicial complex $\mathcal{K}$ with each simplicial complex $\mathcal{K}_i$ is associated with $\epsilon(i)$ in the following way: 
\begin{itemize}
\item For given $i$ and $\epsilon(i)$ any two simplicies $\alpha$, $\beta$ are connected if the pairwise distance of $d(\alpha, \beta)$ between $\alpha$ and $\beta$ is less than or equal to $\epsilon(i)$.

\item When the simplicial complex is built towards $\mathcal{K}_i$, for $\alpha, \beta \in \mathcal{K}_i$, $\alpha \cap \beta$ is either the empty set or a face of both $\alpha$ and $\beta$.  
\end{itemize}

The above filtration procedure is known as the Vietoris-Rips filtration, and the constructed complex is called the Vietoris-Rips complex, $VR(X, \epsilon)$. 
The filtration parameter $\epsilon$ is an increasing function of $i$. 
As an example, consider a point cloud on a plane composed of four points on a unit square, $a, b, c, d$ as shown in Figure~\ref{rips_complex}. 
Suppose that $d(a, b) = d(b, c) = d(c, d) = d(d, a)= 1$ and $d(a, c) = d(b, d) = \sqrt{2}$. 
Figure~\ref{rips_complex} shows the constructed $VR(X, \epsilon)$ for $\epsilon = 0, 1, \sqrt{2}$, from left to right, respectively. 
The figure also shows the corresponding Betti numbers for the zero and one dimensions. 
When $\epsilon = 1$ the one-dimensional hole is created for $\mathcal{K}_{\epsilon = 1}$, but the hole disappears when $\epsilon = \sqrt{2}$ as the created triangle is filled by definition. 
This example shows that the one-dimensional hole persists from $\epsilon = 1$ to $\epsilon = \sqrt{2}$. 
We say that the {\it persistence}, $p$, of the hole is $p = \sqrt{2} - 1$. 

Let $Z^i_n$ and $B^i_n$ be the kernel and image groups of $\mathcal{K}_i$, respectively. Then the $p$-persistent $n$th homology group of $\mathcal{K}_i$ is given by 
\begin{equation*}
H^{i,p}_n(\mathcal{K}_i) = Z^i_n/(B^{i+p}_n \cap Z^i_n). 
\end{equation*}

\begin{figure}[h]
\begin{center}
\begin{tikzpicture}
\filldraw[color=blue!20, fill=blue!20, very thick] (3,0) circle (0.2);
\filldraw[color=blue!20, fill=blue!20, very thick](4,0) circle (0.2);
\filldraw[color=blue!20, fill=blue!20, very thick](4,1) circle (0.2);
\filldraw[color=blue!20, fill=blue!20, very thick](3,1) circle (0.2);
\node at (3.5, -1) (nodeXi) {{\bf $\mathcal{K}_{\epsilon = 0}$}};
\node at (3.5, -2) (nodeXi) {$\beta_0 = 4$};
\node at (3.5, -2.5) (nodeXi) {$\beta_1 = 0$};

\node at (5.5, 0.5) (nodeXi) {$\hookrightarrow$};

\node at (3, 1) (nodeXi) {$a$};
\node at (3, 0) (nodeXi) {$b$};
\node at (4, 0) (nodeXi) {$c$};
\node at (4, 1) (nodeXi) {$d$};


\draw [color = blue!99] (7,0) -- (8,0); 
\draw [color = blue!99] (7,0) -- (7,1); 
\draw [color = blue!99] (7,1) -- (8,1); 
\draw [color = blue!99] (8,1) -- (8,0); 
\filldraw[color=blue!20, fill=blue!20, very thick](7,0) circle (0.2);
\filldraw[color=blue!20, fill=blue!20, very thick](8,0) circle (0.2);
\filldraw[color=blue!20, fill=blue!20, very thick](8,1) circle (0.2);
\filldraw[color=blue!20, fill=blue!20, very thick](7,1) circle (0.2);
\node at (7.5, -1) (nodeXi) {$\mathcal{K}_{\epsilon = 1}$};
\node at (7.5, -2) (nodeXi) {$\beta_0 = 1$};
\node at (7.5, -2.5) (nodeXi) {$\beta_1 = 1$};

\node at (9.5, 0.5) (nodeXi) {$\hookrightarrow$};

\draw [color = blue!99] (11,0) -- (12,0); 
\draw [color = blue!99] (11,0) -- (11,1); 
\draw [color = blue!99] (11,1) -- (12,1); 
\draw [color = blue!99] (12,1) -- (12,0); 
\draw [color = blue!99] (11,0) -- (12,1); 
\path[draw, fill=green!20] (11,0)--(11,1)--(12,1)--cycle;
\path[draw, fill=red!20] (11,0)--(12,0)--(12,1)--cycle;
\filldraw[color=blue!20, fill=blue!20, very thick](11,0) circle (0.2);
\filldraw[color=blue!20, fill=blue!20, very thick](12,0) circle (0.2);
\filldraw[color=blue!20, fill=blue!20, very thick](12,1) circle (0.2);
\filldraw[color=blue!20, fill=blue!20, very thick](11,1) circle (0.2);
\node at (11.5, -1) (nodeXi) {$\mathcal{K}_{\epsilon = \sqrt{2}}$};
\node at (11.5, -2) (nodeXi) {$\beta_0 = 1$};
\node at (11.5, -2.5) (nodeXi) {$\beta_1 = 0$};
\end{tikzpicture}
\end{center}
\caption{An example of filtration over a square point cloud on a plane.}
\label{rips_complex}
\end{figure}
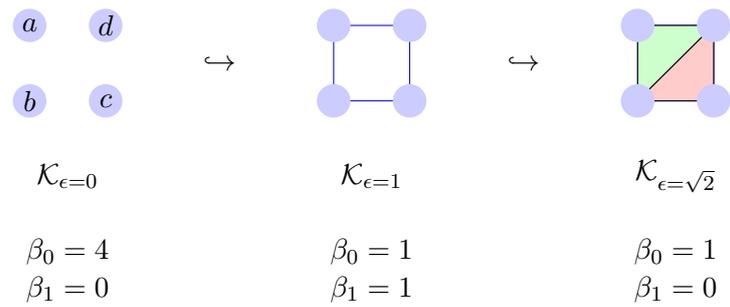

Persistent homology computes homology over the sequence of filtered complexes by computing homology for each $\mathcal{K}_i$ in the sequence. Thus, as the filtration parameter $\epsilon$ increases, persistent homology shows how homology changes over the building procedure. Particularly, as explained in Figure~\ref{rips_complex}, persistent homology reveals the persistence of $n$-dimensional holes, if they exist. As simplices are connected to form a complex, $n$-dimensional holes are created and eventually annihilated, each with a finite persistence $p$.

The collection of such persistence in each dimension with the filtration $\epsilon$ continuously changing from $\epsilon = 0$ can be represented as so-called persistence barcode. Particularly, when an $n$-dimensional hole is created during the filtration process, the corresponding $\epsilon$ value is called the birth time of the hole. Similarly, the corresponding $\epsilon$ value is called the death time when the hole is annihilated. The $n$-dimensional  barcode is the collection of all such birth and death times. Since multiple holes can be created and persist simultaneously, the persistence barcode is a multiset. Persistent barcodes for both the zero and one dimensions of the unit square example above are shown in Figure~\ref{barcode}.  
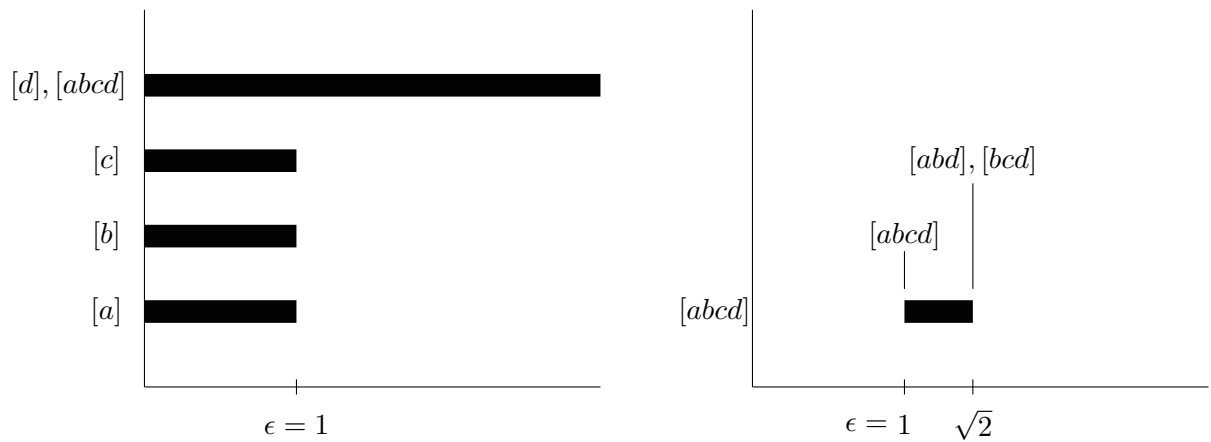
\begin{figure}
\begin{center}
\begin{tikzpicture}

\draw [color = black!99] (1,0) -- (7,0); 
\draw [color = black!99] (1,0) -- (1,5); 

\draw [color = black!99, line width=3mm] (1,1) -- (3,1); 
\draw [color = black!99, line width=3mm] (1,2) -- (3,2); 
\draw [color = black!99, line width=3mm] (1,3) -- (3,3); 
\draw [color = black!99, line width=3mm] (1,4) -- (7,4); 

\draw [color = black!99] (3,-0.1) -- (3,0.1); 
\node at (3,-0.5) (nodeXi) {$\epsilon = 1$};

\node at (0.5,1) (nodeXi) {$[a]$};
\node at (0.5,2) (nodeXi) {$[b]$};
\node at (0.5,3) (nodeXi) {$[c]$};
\node at (0,4) (nodeXi) {$[d], [abcd]$};

\draw [color = black!99] (9,0) -- (15,0); 
\draw [color = black!99] (9,0) -- (9,5); 

\draw [color = black!99] (11,-0.1) -- (11,0.1); 
\node at (10.5,-0.5) (nodeXi) {$\epsilon = $};
\node at (11,-0.5) (nodeXi) {$1$};

\draw [color = black!99, line width=3mm] (11,1) -- (11.9,1); 
\draw [color = black!99] (11.9,-0.1) -- (11.9,0.1); 
\node at (11.9,-0.5) (nodeXi) {$\sqrt{2}$};
\node at (8.5,1) (nodeXi) {$[abcd]$};

\node at (11,2) (nodeXi) {$[abcd]$};
\draw [color = black!99] (11,1.8) -- (11,1.3);

\node at (11.9,3) (nodeXi) {$[abd], [bcd]$};
\draw [color = black!99] (11.9,2.7) -- (11.9,1.3);

\end{tikzpicture}
\end{center}
\caption{Persistence barcode: zero dimension (left) and one dimension (right). }
\label{barcode}
\end{figure} 
The figure shows each persistence generator as a horizontal line. For the one dimension, each generator starts at $\epsilon = 0$, say, for $[a], [b], [c], [d]$ and ends at $\epsilon = 1$, e.g. for $[a], [b], [c]$. The order is not important and could be different. For the one-dimensional case, the generator starts at $\epsilon = 1$ when $[abcd]$ is created and ends at $\epsilon = \sqrt{2}$ when $[abcd]$ breaks into $[abd]$ and $[bcd]$. The starting filtration value of the generator is called the birth time, and the ending filtration value is called the death time. 

Persistence diagram is equivalent to a persistence barcode. Persistence diagram represents each generator in the corresponding persistence barcode as a point in the $xy$ plane with the birth time as the $x$-coordinate and the death time as the $y$-coordinate of the point. Figure~\ref{diagram} shows the persistence diagram of the zero (left) and one (right) dimensions for the unit square data above. For the zero dimension, notice that all three generators coincide at one point. The lasting generator is not displayed in the diagram. 

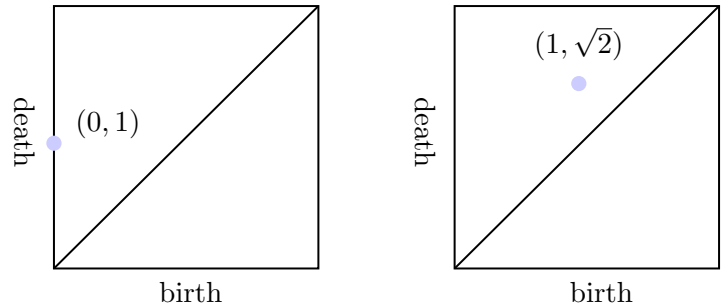
\begin{figure}[h]
\centering
\tikzset{every picture/.style={line width=0.75pt}} 

\begin{tikzpicture}[x=0.75pt,y=0.75pt,yscale=-1,xscale=1]

\draw   (18,31) -- (150,31) -- (150,163) -- (18,163) -- cycle ;
\draw    (150,31) -- (18,163) ;

\draw   (218,31) -- (350,31) -- (350,163) -- (218,163) -- cycle ;
\draw    (350,31) -- (218,163) ;

\filldraw[color=blue!20, fill=blue!20, very thick] (18,100) circle (3);

\filldraw[color=blue!20, fill=blue!20, very thick] (280,70) circle (3);

\node at (45, 90) (nodeXi) {$(0,1)$};
\node at (280, 50) (nodeXi) {$(1, \sqrt{2})$};

\draw (69,168) node [anchor=north west][inner sep=0.75pt]   [align=left] {birth};
\draw (10,75) node [anchor=north west][inner sep=0.75pt]  [rotate=-90] [align=left] {death};

\draw (274,168) node [anchor=north west][inner sep=0.75pt]   [align=left] {birth};
\draw (210,75) node [anchor=north west][inner sep=0.75pt]  [rotate=-90] [align=left] {death};

\end{tikzpicture}
\caption{The persistence diagram of the zero (left) and one (right) dimensions for the unit square data example. }
\label{diagram}
\end{figure}

As shown in the example above, the definition of distance is crucial for persistent homology. Although the Euclidean metric on a plane was used in the above example, the choice of metric depends on the problem considered. In this paper, we will consider graph data $G$, composed of a vertex set $V$ and an edge set $E$, i.e., $G=(V,E)$. Further, we assume that there is no isolated node in $G$. To apply persistent homology to graph data, we need to define the distance between nodes. The distance between two distinct nodes can be defined by using the weight of the edge between them. If every pair of nodes in $G$ is directly connected, the distance can be straightforwardly defined using the edge weights. However, not all nodes are directly connected in general unless $G$ is a complete graph. There are two possible ways to define the distance between pairs of nodes that are not directly connected. The first approach is to create a direct connection between nodes that do not have an edge in the given data $G$. The second approach is path-dependent distances that define the distance based on the path that connects those nodes using the existing intermediate paths (edges) in $G$. Figure~\ref{path-distance} illustrates the first (middle) and second (right) definitions of distance. In the figure, the left figure shows the original graph $G = (V = \{a, b, c, d \}, E = \{ab, ad, dc\})$. The middle figure shows that the distance between $a$ and $c$, $d(a, c)$, is defined by directly connecting $a$ and $c$, which is represented by the red dashed line. The right figure shows that $d(a,c)$ is defined using the edges, $ad$ and $dc$ between $a$ and $c$, which is also represented by the red dashed line. In this paper, we consider the second approach for $d(a, c)$. In the following section, we provide different distance definitions based on the second approach. 

\begin{figure}[h]
\begin{center}
\begin{tikzpicture}

\draw [color = blue!99] (-1,1) -- (0,1); 
\draw [color = blue!99] (-1,1) -- (-1,0); 
\draw [color = blue!99] (0,1) -- (0,0); 

\filldraw[color=blue!20, fill=blue!20, very thick] (-1,0) circle (0.2);
\filldraw[color=blue!20, fill=blue!20, very thick](0,0) circle (0.2);
\filldraw[color=blue!20, fill=blue!20, very thick](0,1) circle (0.2);
\filldraw[color=blue!20, fill=blue!20, very thick](-1,1) circle (0.2);

\node at (-1, 1) (nodeXi) {$a$};
\node at (-1, 0) (nodeXi) {$b$};
\node at (0, 0) (nodeXi) {$c$};
\node at (0, 1) (nodeXi) {$d$};

\draw [color = blue!99] (3,1) -- (4,1); 
\draw [color = blue!99] (3,1) -- (3,0); 
\draw [color = blue!99] (4,1) -- (4,0); 
\draw [dashed, color = red!99] (3,1) -- (4,0);

\filldraw[color=blue!20, fill=blue!20, very thick] (3,0) circle (0.2);
\filldraw[color=blue!20, fill=blue!20, very thick](4,0) circle (0.2);
\filldraw[color=blue!20, fill=blue!20, very thick](4,1) circle (0.2);
\filldraw[color=blue!20, fill=blue!20, very thick](3,1) circle (0.2);

\node at (3, 1) (nodeXi) {$a$};
\node at (3, 0) (nodeXi) {$b$};
\node at (4, 0) (nodeXi) {$c$};
\node at (4, 1) (nodeXi) {$d$};

\draw [dashed, color = red!99] (7,1.1) -- (8,1.1); 
\draw [dashed, color = red!99] (8.1,1) -- (8.1,0); 


\draw [color = blue!99] (7,0) -- (7,1); 
\draw [color = blue!99] (7,1) -- (8,1); 
\draw [color = blue!99] (8,1) -- (8,0); 
\filldraw[color=blue!20, fill=blue!20, very thick](7,0) circle (0.2);
\filldraw[color=blue!20, fill=blue!20, very thick](8,0) circle (0.2);
\filldraw[color=blue!20, fill=blue!20, very thick](8,1) circle (0.2);
\filldraw[color=blue!20, fill=blue!20, very thick](7,1) circle (0.2);

\node at (7, 1) (nodeXi) {$a$};
\node at (7, 0) (nodeXi) {$b$};
\node at (8, 0) (nodeXi) {$c$};
\node at (8, 1) (nodeXi) {$d$};

\end{tikzpicture}
\end{center}
\caption{Distance between nodes. Left: the original graph $G$. Middle: $d(a,c)$ is defined by creating a new edge between $a$ and $c$. Right: $d(a,c)$ is defined based on the existing intermediate paths between $a$ and $c$.}
\label{path-distance}
\end{figure}
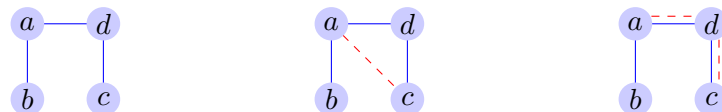

A common approach to defining distance in a given graph $G$ based on intermediate paths or edges is to utilize the weights of the edges, provided they are defined over $G$. The edge weights are problem-specific. For example, in the case of music data, as discussed in~\citep{Bergomithesis,Bigo,BWold,Tran02012024,Tran02092023}, distinct tuples consisting of the pitch and duration of each musical note are treated as nodes. If two nodes appear consecutively in the music, they are connected by an edge, with the edge weight defined by the frequency of their co-occurrence. This representation allows the given music to be modeled as a graph. Now, consider two nodes that are directly connected by an edge, meaning they appear consecutively in the music. A natural way to define the distance between these nodes is as the reciprocal of the weight, i.e., the reciprocal of the frequency. This definition is intuitive because nodes that appear together more frequently are considered closer. For nodes that are not directly connected, this reciprocal frequency-based distance can be generalized. However, such a generalization is not unique, as different approaches can be taken. This paper explores three possible definitions of distance in this context.

\section{Music graph networks and metric definitions} 

In this paper, we analyze how variations in distance definitions affect persistent homology. 
We begin by introducing three different distance definitions and then examine the resulting changes in persistent homology with associated properties.

\subsection{A construction of music graph network}
For simplicity, and for the numerical examples presented in later sections, we consider monophonic music. However, this idea can also be generalized to polyphonic music.

Let $G = \{V, E\}$ be the graph constructed from the given music data $T$. Here $T$ is a time series, $T = \{ T_1, T_2, \cdots, T_N\}$. We define the nodes and edges of $G$ as follows:
\begin{definition}
Each node, $n \in V$, in $G$ is defined as a pair consisting of a pitch and its duration, represented as the tuple:
\begin{equation}
n = (\mbox{pitch}, \mbox{duration}).
\nonumber
\end{equation}
\end{definition}
Note that the above node definition is not unique. This definition can be extended to polyphonic music, for example, by representing nodes as sets of pitches or by projecting them onto a basic scale~\citep{gomez}. We then define each edge $e_{ij} \in E$ as follows:
\begin{definition}
Let $e_{ij} \in E$ in $G$ be the edge connecting two distinct nodes $n_i$ and $n_j$ ($n_i \ne n_j$ for $i \ne j$) when $n_i$ and $n_j$ appear consecutively in music.
\end{definition}
Note that in this definition, we do not consider directionality. That is, the graph $G$ constructed from the given music data is an undirected graph. 
Furthermore, we consider the graph $G$ to be weighted. The definition of edge weights is not unique, and one possible approach is to use frequency. The weights associated with edges based on frequency are defined as follows:
\begin{definition}
Let $w_{ij}$ be the weight of $e_{ij} \in E$ in $G$, $i \ne j$. The frequency-based weight $w_{ij}$ is defined as
 \begin{equation*}
   w_{ij} = | \{(T_k, T_{k+1})| (T_k, T_{k+1})  = (n_i, n_j) \mbox{ or } (n_j, n_i), i\ne j, k = 1, \cdots, N\} | ^{-1}. 
 \end{equation*}
\end{definition}
Also, note that the definition of weights may vary depending on the specific problem being considered.

\subsection{Three path-based distance definitions}
Let $M = (V, E, W)$ be a music network represented as a weighted graph. Assume that $M$ is connected. Let $W$ be the weight function, defined as $W{:}\; E \rightarrow \mathbb{R}$.

For any two distinct nodes $v, w \in V$, there always exists a path $p$ connecting $v$ and $w$ in $M$. For such a path $p$, we consider the notion of the \textit{shortest path}. There are multiple ways to define the shortest path.

First, we define the shortest path in terms of the minimum number of edges in $p$.
Clearly, such paths are not necessarily unique. However, if we impose an ordering on the nodes, say $v_1 < v_2 < \dots < v_n$, where $n$ is the total number of distinct nodes in $M$, we can determine a unique path with the minimal number of edges using the unweighted Dijkstra algorithm. Thus, this choice of a unique minimal path depends on the vertex ordering. We refer to this path as {the minimal edge path}.

Second, we consider the path that minimizes the sum of the weights of all the edges in $p$.
Since multiple such paths may exist, we again use the vertex ordering to define a unique path via the weighted Dijkstra algorithm. We refer to this path as {the minimal weight path}.

\begin{definition}[$d_1(v,w)$: The first definition of distance]
    For nodes $v,w$ in $M$, let $p_{vw}$ be the minimal edge path.
    Define the distance $d_1: V \times V \rightarrow \mathbb{R}$ to be 
    \begin{equation*}
    d_1(v,w) =\sum_{\substack{e \in p_{vw} \\ e \text{: edge}}} W(e) \text{ if } v \neq w 
    \end{equation*} and $d_1(v,w) = 0$ if $v=w$.
\end{definition}

\begin{definition}[$d_2(v,w)$: The second definition of distance]
 For nodes $v,w$ in the network, let $p_{vw}$ be the minimal weight path.
 Define the distance $d_2: V \times V \rightarrow \mathbb{R}$ to be 
 \begin{equation*}
 d_2(v,w) = \sum_{\substack{e \in p_{vw} \\ e \text{: edge }}} W(e)\text{ if } v \neq w 
 \end{equation*} and $d_2(v,w) = 0$ if $v=w$.
\end{definition}

\begin{definition}[$d_3(v,w)$: The third definition of distance]
 For nodes $v,w$ in the network, let $\mathcal{P}(v,w)$ be the set of paths with the minimum number of edges between $v$ and $w$. 
 Define the distance $d_3: V \times V \rightarrow \mathbb{R}$ as 
 \begin{equation*}
 d_3(v,w) = \min_{p \in \mathcal{P}(v,w) } \left\{ \sum_{\substack{e \in p \\ e \text{: edge }}} W(e) \right\}\text{ if } v \neq w 
 \end{equation*} and $d_3(v,w) = 0$ if $v=w$.
\end{definition}
The third definition of distance can be regarded as the combination of Definitions $1$ and $2$.

Figure~\ref{fig:music} shows a segment of the Korean traditional music piece {\it Sanghyeondodeuri}\footnote{{\it Sanghyeondodeuri}, a part of {\it Junggwangjigok}, is a piece of traditional Korean court music that belongs to the Dodeuri genre, which features repetitive melodic patterns and was performed in royal and aristocratic settings. 
It is characterized by its elegant, flowing structure and is typically played with traditional Korean instruments such as the geomungo, gayageum, and daegeum.}, performed by five different instruments. Figure~\ref{fig:distance} shows the distance matrix $D$, $D_{vw} = d(v,w)$, computed using different distance definitions, $d_1, d_2, d_3$, for the geomungo\footnote{The geomungo is a traditional Korean zither-like string instrument with six strings and 16 convex frets, played using a bamboo plectrum called a suldae. Known for its deep, resonant tones, it has been historically favored by scholars and aristocrats for its rich, expressive sound and versatility in both solo and ensemble performances.} instrument part in {\it Sanghyeondodeuri}. 
By definition, all elements of the distance matrix are non-negative. The figure shows the contour of the distance, with red representing the maximum value and blue the minimum.
It shows that the distances are different depending on the choice of the distance function. In the figure, we clearly observe the ordering relationships among the different distance definitions of  $d_2(v, w) \le d_3(v, w) \le d_1(v, w)$, which will be discussed in the next section. In fact, this ordering relationship is not confined to specific musical pieces but applies universally to all music.

\begin{figure}[!ht]
\begin{center}
\includegraphics[width=0.7\textwidth]{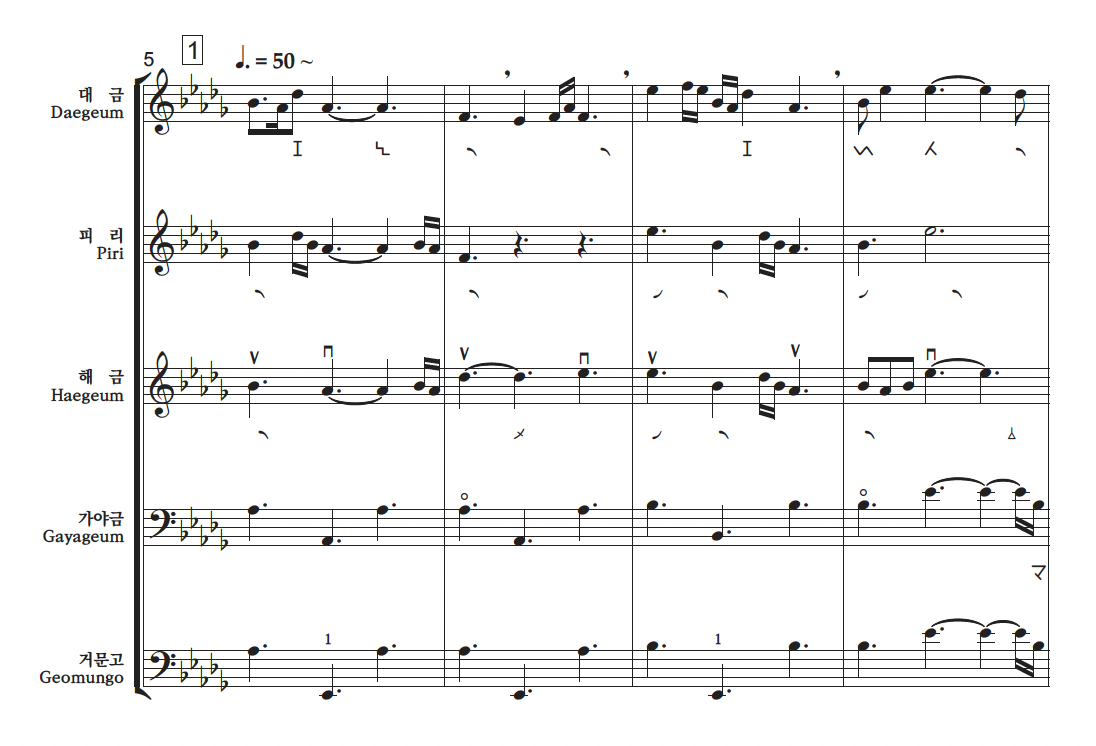}
\caption{A segment of the Korean traditional music piece {\it Sanghyeondodeuri} with five instruments. }
\label{fig:music}
\end{center}
\end{figure}

\begin{figure}[!ht]
\begin{center}
\includegraphics[width=1.0\textwidth]{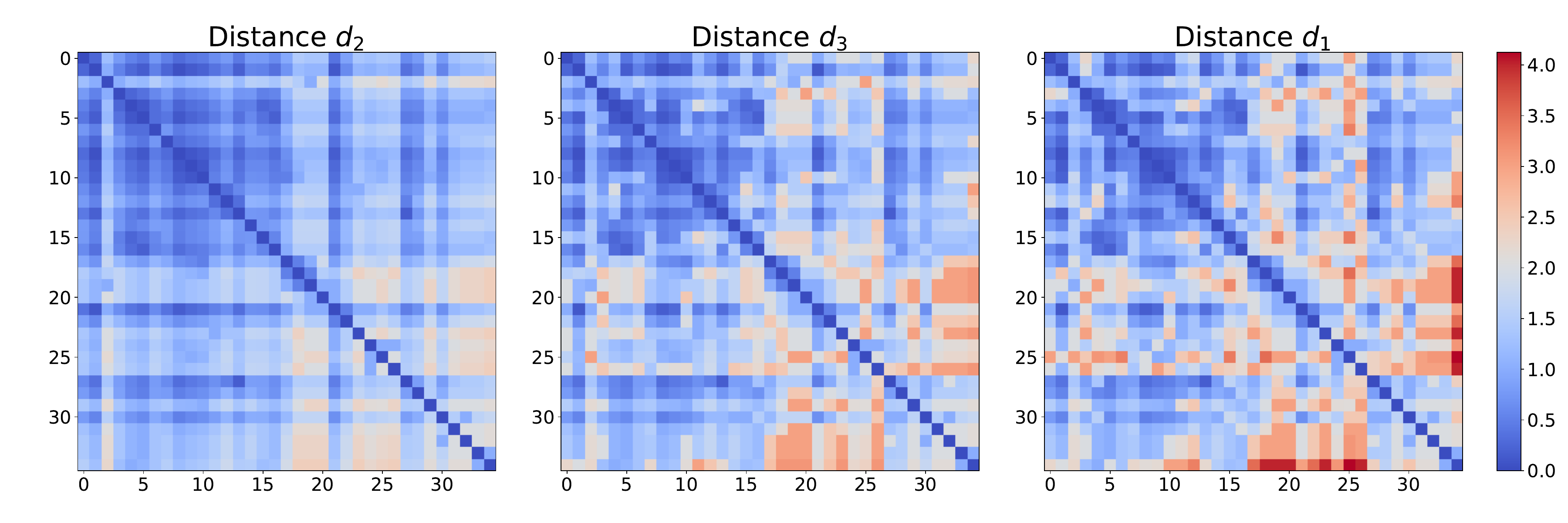}
\caption{Distance matrices corresponding to $d_2$ (left), $d_3$ (middle) and $d_1$ (right) of {\it Sanghyeondodeuri} for geomungo part. 
Notice the ordering relationships among the different distance definitions of  $d_2(v, w) \le d_3(v, w) \le d_1(v, w)$.
}
\label{fig:distance}
\end{center}
\end{figure}

\subsection{Inequalities among three definitions} 
We can also consider the reciprocal version of the distance for each definition.
For example, we can explicitly define the reciprocal version of Definition $1$ as the following.
If the weight function is non-zero at every edge, we can define the reciprocal distance $d_1^{R}: V \times V \rightarrow \mathbb{R}$ by 
 \begin{equation*}
 d_1^{R}(v,w) = \sum_{\substack{e \in p_{vw} \\ e \text{: edge }}} \frac{1}{W(e)},
 \end{equation*}
where $p_{vw}$ is the minimal edge path connecting $v$ and $w$.

\begin{proposition}
    Let $M=(V,E,W)$ be the music network and suppose that $W$ is positive.
    Denote $W^{-1} = \frac{1}{W} $ by the reciprocal weight function of $W$.
    If we define the reciprocal music network $M^{R}$ as $(V,E,W^{-1})$, then we have 
    \begin{equation*} 
    d_1^{R} \text{ on } M \equiv d_1 \text{ on } M^{R}.  
    \end{equation*}
\end{proposition}

Actually, the property could be generalized into the compatibility with the composition.
Let $f {:}\; \mathbb{R} \rightarrow \mathbb{R}$ be the arbitrary real-valued function. 
Let $\mathcal{F}$ be the composition operator with $f$ defined to be $\mathcal{F}(W) = f \circ W$, where $W$ is any function from $E$ to $\mathbb{R}$.
Then $\mathcal{F}$ becomes a function from $F(E,\mathbb{R})$ to itself, where $F(E,\mathbb{R})$ denotes the set of functions from $E$ to $\mathbb{R}$.
Let $d {:}\; V \times V \rightarrow \mathbb{R}$ be the distance.
For two distinct nodes $v,w$ in $M$, let $p_{vw}$ be the minimal edge path. 
Define the composed distance $\mathcal{F}(d)$ to be 
\begin{equation*}
\mathcal{F}(d)(v,w) = \sum_{\substack{e \in p_{vw} \\ e \text{: edge }}} \mathcal{F}(W)(e) \text{ if } v \neq w 
\end{equation*} and $\mathcal{F}(d)(v,w) = 0$ if $v=w$.
Then the following property holds:

\begin{proposition}
    Let $M=(V,E,W)$ be the music network and $\mathcal{F}$ be the composition operator with $f:\mathbb{R} \rightarrow \mathbb{R}$.
    Suppose that $\mathcal{F}(W)$ is well-defined and consider the composition $\mathcal{F}(M) = (V,E,\mathcal{F}(W))$, the composed music network.
    Then we have a congruence 
    \begin{equation*}
    \mathcal{F}(d_1) \text{ on } M \equiv d_1 \text{ on } \mathcal{F}(M). 
    \end{equation*}
\end{proposition}

\begin{proof}
    Let $\mathcal{F}$ be a composition operator with $f:\mathbb{R} \rightarrow \mathbb{R}$.
    Take any two nodes $v,w$ in $M=(V,E,W)$.
    Then the shortest path $p_{vw}$ in $M$ and $\mathcal{F}(M)$ is exactly the same because $p_{vw}$ is determined by the graph $G=(V,E)$.
    The Dijkstra algorithm only depends on the ordering of the set of vertices, which remains the same even after the weight changes. 
    Therefore, we have 
    \begin{equation*}
    \mathcal{F}(d_1)(v,w) \text { on } M = \sum_{\substack{e \in p_{vw} \\ e \text{: edge }}} \mathcal{F}(W)(e) = \sum_{\substack{e \in p_{vw} \\ e \text{: edge }}} (f \circ W)(e) = d_1 (v,w) \text { on } \mathcal{F}(M). 
    \end{equation*}
\end{proof}

This means that the distance is independent of the order of action over the weight function for Definition $1$.
Therefore, $d_1$ can be readily computed when the weights are modified.
In the case of $d_2$ and $d_3$, Proposition $2$ does not hold.

\begin{figure}[!ht]
\centering
\begin{tikzpicture}[auto=left]
  
  \node[circle,fill=blue!20] (n1) at (2,10) {a};
  \node[circle,fill=blue!20] (n2) at (1,8)  {b};
  \node[circle,fill=blue!20] (n3) at (4,6)  {c};
  \node[circle,fill=blue!20] (n4) at (6,8) {d};

  \foreach \from/\to/\weight in {n1/n2/1, n2/n3/1, n3/n4/1, n4/n1/0.5}
    \draw (\from) -- node[midway, above, sloped] {\weight} (\to);

\end{tikzpicture}
\caption{The graph given in Example ~\ref{def:ex1}.}
\label{fig:musicnetwork1}
\end{figure}
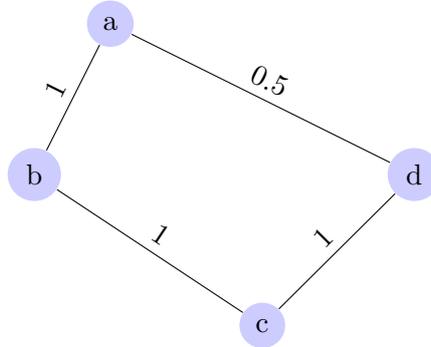

\begin{example}\label{def:ex1}
 Let $V= \{a,b,c,d \}$ be the set of nodes and $E = \{ \{a,b\},\{b,c\},\{c,d\},\{a,d\} \}$ be the set of edges. Define the weight function $W: E \rightarrow \mathbb{R}$ to be 
 \begin{equation*}
 W(e)= \begin{cases}
     0.5 & \text{ if } e=\{ a,d \} \\
     1 & \text{ otherwise }.\\
    \end{cases}
 \end{equation*}
Define the music network $M=(V,E,W)$ as shown in Figure~\ref{fig:musicnetwork1}.
Then $d_2(a,c) = 1 + 1 = 2$ on $M^{R}$ because $abc$ is the minimal weight path from $a$ to $c$ in $M^{R}$.
However, the minimal weight path from $a$ to $c$ in $M$ is $adc$.
Therefore, $d_2^{R}(a,c) = \frac{1}{1} + \frac{1}{0.5} = 3$ on $M$.
Hence $d_2(a,c) \text{ on } M^{R} \neq d_2^{R}(a,c) \text{ on } M $.
Likewise $d_3(a,c) \text{ on } M^{R} \neq d_3^{R}(a,c) \text{ on } M $ by the same reason.
\end{example}

Next, we check whether each definition satisfies the metric condition. In conclusion, $d_1$ and $d_3$ do not satisfy the triangle inequality while $d_2$ does.

\begin{proposition}
    Let $d_2$ be the distance defined in Definition $2$.
    Then $d_2$ is a metric. That is, $d_2$ satisfies the triangle inequality.
\end{proposition}

\begin{proof}
    First, for any node $v$ in $M$, we know that $d_2(v,v) = 0$ by the definition of $d_2$.
    Second, take any two nodes $v,w$ in $M$.
    In Dijkstra's algorithm, it is important to note that the process of finding a path between vertices $v$ and $w$ may yield different results compared to finding a path between $w$ and $v$. Nevertheless, it is worth noting that for all paths, the sum of the weights along the path remains unique, as it represents the minimum total weight among all possible paths that traverse the edges. Therefore, we know that $d_2(v,w)=d_2(w,v)$. Finally, take any three distinct nodes $v,z,w$ in $M$.
    Denote the minimum weight path from $v$ to $z$ by $p_{vz}$ and the minimum weight path from $z$ to $w$ by $p_{zw}$.
    Here, the concatenated path $p_{vz} + p_{zw}$ is also a path between $v$ and $w$.
    Therefore, we have
    \begin{equation*} 
    d_2(v,w) = \sum_{\substack{e \in p_{vw} \\ e \text{ : edge }}} W(e) 
    \leq \sum_{\substack{e \in p_{vz} + p_{zw} \\ e \text{ : edge }}} W(e) 
    = d_2(v,z) + d_2(z,w)
    \end{equation*}
    
\end{proof}

However, in the following example, $d_1$ and $d_3$ do not become metrics.

\begin{example}\label{def:ex2}

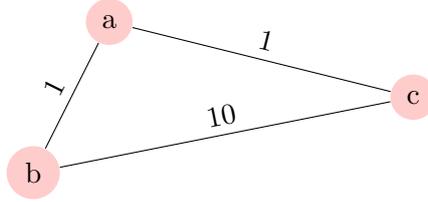
\begin{figure}[!ht]
\centering
\begin{tikzpicture}[auto=left]
  
  \node[circle,fill=red!20] (n1) at (2,10) {a};
  \node[circle,fill=red!20] (n2) at (1,8)  {b};
  \node[circle,fill=red!20] (n3) at (6,9)  {c};

  \foreach \from/\to/\weight in {n1/n2/1, n2/n3/10, n3/n1/1}
    \draw (\from) -- node[midway, above, sloped] {\weight} (\to);

\end{tikzpicture}
\caption{The graph given in Example ~\ref{def:ex2}.}
\label{fig:musicnetwork2}
\end{figure}

 Let $V= \{a,b,c\}$ be the set of nodes and $E = \{ \{a,b\},\{b,c\},\{a,c\} \}$ be the set of edges. Define the weight function $W: E \rightarrow \mathbb{R}$ to be 
 \begin{equation*}
 W(e)= \begin{cases}
     10 & \text{ if } e=\{ b,c \} \\
     1 & \text{ otherwise } . 
    \end{cases}
 \end{equation*}
Define the network $M=(V,E,W)$ as shown in Figure~\ref{fig:musicnetwork2}.
Then 
\begin{equation*}
d_1(b,c)=10 > 1+1 = d_1(a,b) + d_1(a,c) \text{ on } M.
\end{equation*}
Also,  
\begin{equation*}
d_3(b,c)=10 > 1+1 = d_3(a,b) + d_3(a,c) \text{ on } M.
\end{equation*}

However, since the minimal weight path from $b$ to $c$ in $M$ is $bac$, so we have 
\begin{equation*}
d_2(b,c)=2 \leq 1+1 = d_2(a,b) + d_2(a,c).
\end{equation*}
\end{example}

We can observe that there always exists the following relationship among  $d_1$, $d_2$, and $d_3$.

\begin{proposition}
    Let $d_1$, $d_2$, and $d_3$ be the distances defined above.
    Then the following inequality holds  
    \begin{equation*} 
    d_2(v,w) \leq d_3(v,w) \leq d_1(v,w) \text{ for any } v,w \in V. 
    \end{equation*}
\end{proposition}

\begin{proof}
    Fix the music network $M=(V,E,W)$ and take any pair of two nodes $v$ and $w$ in $V$.
    First, we show that $d_3 \leq d_1 $. 
    Let $\mathcal{P}(v,w)$ be the set of paths with the minimum number of edges between $v$ and $w$. If we define $p_{vw}$ to be the unique minimal edge path found by the Dijkstra algorithm, then $p_{vw} \in \mathcal{P}(v,w)$.
    By definition, we have 
    \begin{equation*}
    d_3(v,w) = \min_{p \in \mathcal{P}(v,w) } \left\{ \sum_{\substack{e \in p \\ e \text{ : edge }}} W(e) \right\} \leq \sum_{\substack{e \in p_{vw} \\ e \text{ : edge }}} W(e) = d_1(v,w). 
    \end{equation*}

    Second, we show that $d_2 \leq d_3$.
    This follows from the relation among the paths between $v$ and $w$.
    Explicitly, let $\mathcal{P}(v,w)$ be defined as above and let $\mathcal{P}_2(v,w)$ be the set of all paths connecting $v$ and $w$. Then it is trivial that  $\mathcal{P}(v,w) \subseteqq \mathcal{P}_2(v,w)$. Hence, we have the following inequality 
    \begin{equation*}
    d_2(v,w) = \min_{p \in \mathcal{P}_2(v,w)} \left\{ \sum_{\substack{e \in p \\ e \text{ : edge }}} W(e) \right\} \le \min_{p \in \mathcal{P}(v,w) } \left\{ \sum_{\substack{e \in p \\ e \text{ : edge }}} W(e) \right\} = d_3(v,w).
    \end{equation*}
    
\end{proof}

\begin{figure}[!ht]
    \centering
    \begin{minipage}{0.45\textwidth}
        \centering
        \begin{tikzpicture}[auto=left]
          
          \node[circle,fill=green!20] (n1) at (4,10) {0};
          \node[circle,fill=green!20] (n2) at (4,8)  {3};
          \node[circle,fill=green!20] (n3) at (8,8)  {2};
          \node[circle,fill=green!20] (n4) at (8,10) {1};
        
          \foreach \from/\to/\weight in {n1/n2/1, n2/n3/2, n3/n4/10, n4/n1/20}
            \draw (\from) -- node[midway, above, sloped] {\weight} (\to);
            \draw[thick,->] (n1) -- (n4);
            \draw[thick,->] (n4) -- (n3);
        \end{tikzpicture}
        \begin{tikzpicture}
            \draw[thick, ->] (0,0) -- (0.5,0) node[right]{Path: 0 - 1 - 2};
        \end{tikzpicture}
    \end{minipage}\hfill
    \begin{minipage}{0.45\textwidth}
        \centering
        \begin{tikzpicture}[auto=left]

          \node[circle,fill=green!20] (n1) at (4,10) {0};
          \node[circle,fill=green!20] (n2) at (4,8)  {1};
          \node[circle,fill=green!20] (n3) at (8,8)  {2};
          \node[circle,fill=green!20] (n4) at (8,10) {3};
        
          \foreach \from/\to/\weight in {n1/n2/1, n2/n3/2, n3/n4/10, n4/n1/20}
            \draw (\from) -- node[midway, above, sloped] {\weight} (\to);
            \draw[thick,->] (n1) -- (n2);
            \draw[thick,->] (n2) -- (n3);
        \end{tikzpicture}
        \begin{tikzpicture}
            \draw[thick, ->] (0,0) -- (0.5,0) node[right]{Path: 0 - 1 - 2};
        \end{tikzpicture}
    \end{minipage}
    \caption{ (Example~\ref{ordering_example}) For the ordering in the left figure, $d_1(0,2)=30$, and for the ordering in the right figure, $d_1(0,2)=3$.}
        \label{fig:my_graph3}
\end{figure}
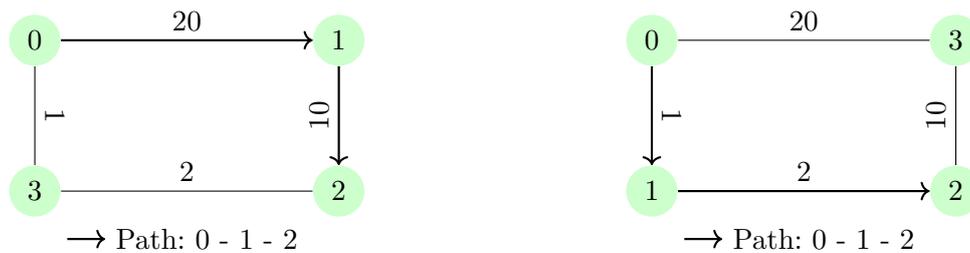

One notable feature is that the music network equipped with the definition of $d_1$ tends to be greatly affected by vertex ordering in the Dijkstra algorithm.  On the other hand, the ones with $d_2$ and $d_3$ are determined independently of vertex ordering.  For example, consider the following example. 
\begin{example}\label{ordering_example}
Suppose that we have a graph with $V(G)=\{v_0, v_1, v_2, v_3 \}$ as illustrated in Figure~\ref{fig:my_graph3}. Assume that in applying Dijkstra's algorithm, the vertices are ordered as  $0 < 1 < 2 < 3$ in $G$. In this case, $d_1$ selects different paths in the scenarios shown on the left and right figures of Figure~\ref{fig:my_graph3}, resulting in different distances. However, in the case of $d_2$ and $d_3$, the distance is consistently calculated as a constant value regardless of the vertex ordering, since both definitions identify the path that minimizes the total weight independently of the vertex order.
\end{example}

There are several special cases that the music graph $G$ has uniform property. For example, if the music structure is strictly non-repetitive and contains no cycles, the music graph becomes a tree and obviously $d_1(v,w) = d_2(v, w) = d_3(v, w)$, for any vertices $v$ and $w$ in $G$. This also holds when $G$ is a complete graph. Also for a cycle graph $G$ of order $n\ge 3$,  $d_1(v, w) = d_3(v, w)$ if $n$ is odd and $d_1(v,w) = d_2(v,w) = d_3(v, w)$ if $G$ is uniformly weighted. 

Such special cases may arise within local graphs or subgraphs of the entire structure, making them particularly relevant when analyzing specific segments of a musical piece for interpretation.
\section{Persistent homology induced by three different distances } 

In this section, we analyze the homological properties of the target graph according to each distance. Particularly, we focus on the inclusion relations of persistence barcodes and diagrams of music networks induced by the three different distance definitions.  

\subsection{Injection relations among three different distances}

Before stating the main theorem of this work, we need to clarify the notations. Recall that the $p$-dimensional persistence barcode with the distance $d$,  $\mathsf{bcd}_p(d)$, is defined as 
$ A \cup B$
where 
\begin{equation*}
A = \{ [f(\sigma_i), f(\sigma_j)]  \mid \sigma_i \text{ is a $p$-simplex with }i = low_R(j), f(\sigma_i) < f(\sigma_j) \}, 
\end{equation*}
and 
\begin{equation*}
B = \{ [f(\sigma_i), \infty]  \mid \sigma_i \text{ is a $p$-simplex, there is no $\sigma_j$ such that } i = low_R(j)\}. 
\end{equation*}
Here we denote the set-version barcode by explicating the corresponding simplices, i.e.
\begin{equation*} 
\mathsf{bcd}^{set}_p(d) \coloneqq \{[b_i,d_j]_{\sigma_i} \mid [b_i,d_j] \in \mathsf{bcd}_p(d) \}.
\end{equation*}
$R$ in $low_R(j)$ is the reduced boundary matrix derived from the original boundary matrix of a filtered simplicial complex.
Each column $j$ of $R$ corresponds to a simplex $\sigma_j$, and reduction is typically performed column by column (usually left to right) to track how simplices pair up. $low_R(j) = i$ means that the pivot (i.e., lowest non-zero entry) in column $j$ is at row $i$. This indicates that $\sigma_i$ and $\sigma_j$   form a persistence pair, where $\sigma_i$ is a $p$-dimensional simplex that creates a homological feature and $\sigma_j$  is a 
$(p+1)$-dimensional simplex that annihilates the same feature.
To make the corresponding simplices well-defined, we fix the way to match such a simplex using the Standard Algorithm~\footnote{See the following for more detailed algorithms regarding $R$: Edelsbrunner, H., Harer, J. (2010), Computational Topology: An Introduction. American Mathematical Society, Zomorodian, A., Carlsson, G. (2005),  Computing persistent homology, Discrete \& Computational Geometry, 33(2), 249--274.}.
By fixing the method, we obtain the useful results as below. The complete proofs are provided, with more general definitions of distance, in~\citep{Heo2024}. 

\begin{proposition}\label{prop1}
    Let $d_i$, $i = 1, 2, 3,$ be the three distances for the given $G$. Then every birth edge of $\mathsf{bcd}^{set}_1(d_i)$ exists in $G$. Therefore, it is justified to call a birth $1$-simplex a birth edge.
\end{proposition}
\begin{proposition}\label{prop2}
    Let $G=(V,E,W_E)$ be a connected weighted graph. 
    If $\mathcal{B}_i$ are the set of all birth edges of $\mathsf{bcd}_1(d_i)$, respectively for $i=1,2,3$, then $\mathcal{B}_2 \subseteq \mathcal{B}_3 \subseteq \mathcal{B}_1$.
\end{proposition}
These propositions also act as lemmas for the following main theorem:
\begin{theorem}\label{maintheorem}
    Let $G=(V,E,W_E)$ be a connected weighted graph. 
    Then there are injections $\varphi_{2,3} : \mathsf{bcd}^{set}_1(d_2) \hookrightarrow \mathsf{bcd}^{set}_1(d_3)$ and $\varphi_{3,1} : \mathsf{bcd}^{set}_1(d_3) \hookrightarrow \mathsf{bcd}^{set}_1(d_1)$ defined as $\varphi_{2,3}([a,b]_{\sigma})=[a,c]_{\sigma} $ and $\varphi_{3,1}([a,c]_{\sigma})=[a,d]_{\sigma} $ such that $b \le c \le d$ for any $[a,b] \in \mathsf{bcd}_1(d_2)$.
\end{theorem}
The above properties lead to a more convenient and solid explanation for the process of TDA.
By Proposition~\ref{prop1} and Theorem~\ref{maintheorem}, we can simultaneously analyze the elements of the persistence barcode with the choice of distances. We can group the elements by the injections $\varphi_{2,3}$ and $\varphi_{3,1}$, which can be found by the birth edge.
The elements in the persistence diagrams can be categorized into three main types, as illustrated in 
Figure~\ref{fig:pdanalysis_combined}. 

\begin{figure}[!ht]
\centering
\includegraphics[width=0.3\textwidth]{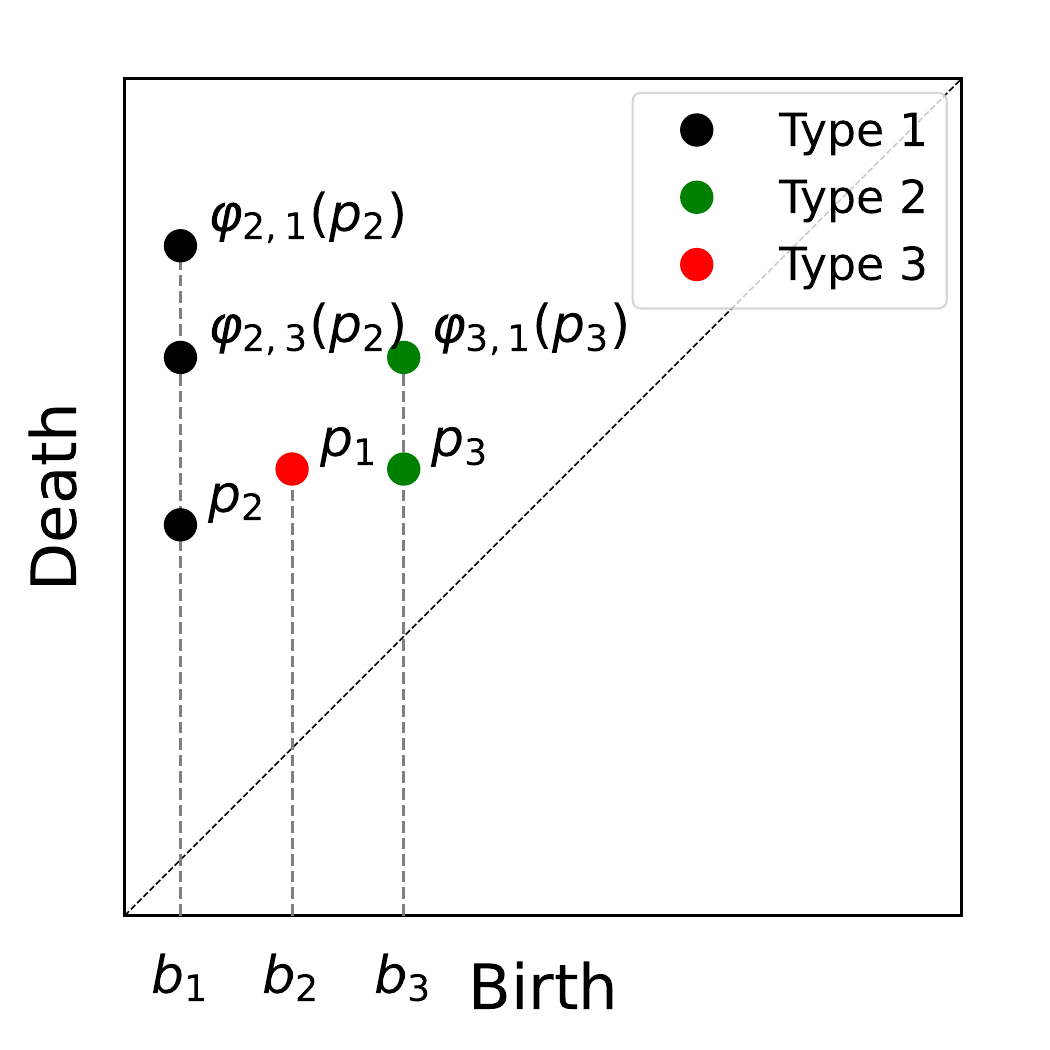}
\caption{Three types of persistence barcodes.}
\label{fig:pdanalysis_combined}
\end{figure}
Figure~\ref{fig:pdanalysis_combined} shows the overlap of persistence diagrams from each definition, aligned to the same scale. The overlapping patterns can be classified into three cases (Type 1 to Type 3): Type 1 (black) indicates that when an element from the $d_2$ definition appears, two corresponding elements from the $d_3$ and $d_1$ definitions also appear. Type 2 (green) indicates that when an element from the $d_3$ definition appears, one corresponding element from the $d_1$ definition is also present. Type 3 (red) represents cases with only a single element from the $d_1$ definition. Note that all elements in each case are aligned to the same birth value.

Figure~\ref{fig:pdanalysis} shows the separation of all elements from Figure~\ref{fig:pdanalysis_combined} into individual persistence diagrams generated by the definitions of $d_1, d_2$, and $d_3$. The diagrams on the left, middle, and right correspond to those generated by the definitions of $d_2$, $d_3$, and $d_1$, respectively. In the left diagram, there is one element (black) induced by $d_2$, which in turn induces the black-marked elements in the middle and right diagrams. 
The element induced by $d_3$ (green) in the middle diagram induces the green-marked element in the right diagram. 
In contrast, elements induced solely by $d_1$ appear only in the $d_1$ persistence diagram.

\begin{figure}[!ht]
\centering
\includegraphics[width=0.8\textwidth]{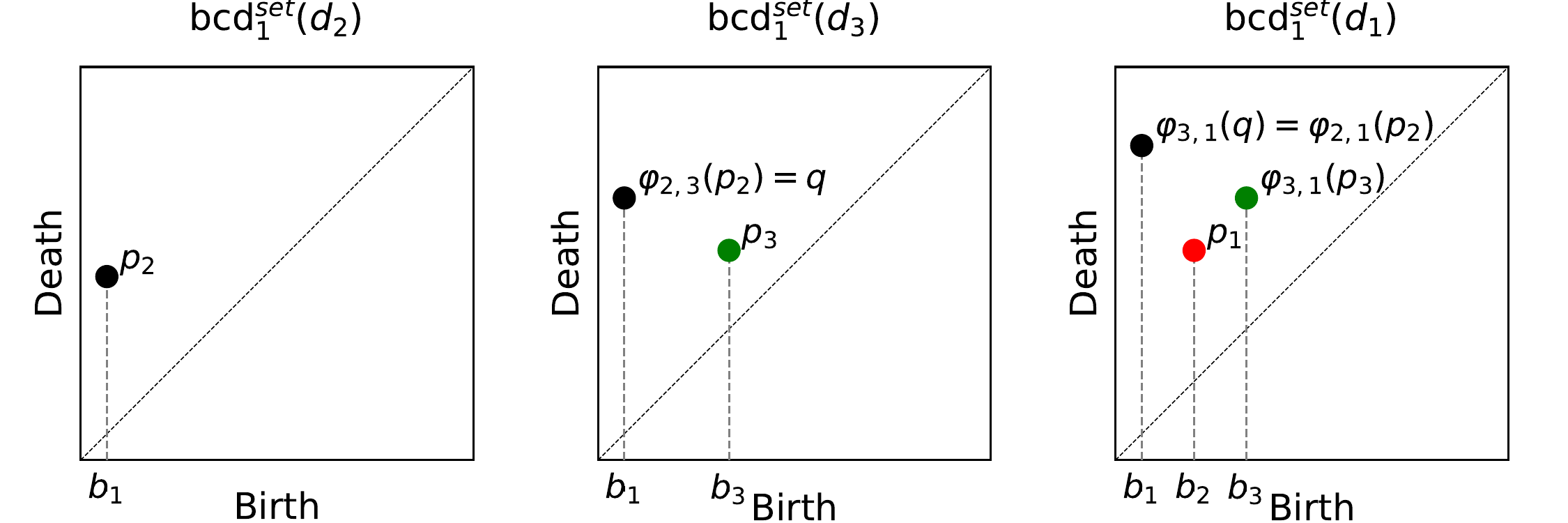}
\caption{Persistence barcodes classified based on Theorem~\ref{maintheorem} and their appearance in $\mathsf{bcd}^{set}_1(d_2)$, $\mathsf{bcd}^{set}_1(d_3)$, and $\mathsf{bcd}^{set}_1(d_1)$.
}
\label{fig:pdanalysis}
\end{figure}
\subsection{Examples}
In this section, we use real music examples to validate the findings presented in the previous sections. The data for the Korean music pieces used in Figures~\ref{fig:heatmap} and \ref{fig:bcNpd}, as well as Table~\ref{tab:compare}, can be found in \cite{data1,data2}.

{\color{black}{Figure~\ref{fig:heatmap} shows the heatmap of the distance matrices corresponding to $d_2$ (left), $d_3$ (middle), and $d_1$ (right) for two original songs in {\it Junggwangjigok}
\footnote{{\it Junggwangjigok}  is a traditional Korean court music piece belonging to the ``Hyangak" (native Korean music) repertoire.
{\it Junggwangjigok}, also known as {\it Yeongsanhoesang}, is a suite of nine instrumental pieces often played at noble gatherings.}
:{\it Sangnyeongsan}
\footnote{{\it Sangnyeongsan}, a part of {\it Junggwangjigok}, is one of the most representative pieces in Korean court music, specifically from the {\it Yeongsanhoesang} suite, which is often performed in elegant and meditative settings. Characterized by its slow tempo and lyrical melody, {\it Sangnyeongsan} serves as the opening movement of the suite and sets a calm, reflective tone.}
(geomungo part) (Top),  and {\it Sanghyeondodeuri} (geomungo part) (Bottom). 
The color of each cell represents the value of its corresponding entry, with red cells indicating higher values and blue cells indicating lower values. 
The diagonal elements of the distance matrix are zero, as indicated by the blue cells, according to the definition of the pairwise distance between two points for all distances $d_1, d_2$, and $d_3$. Also, it is clear from the figure that $d_2(v,w) \leq d_3(v,w) \leq d_1(v,w)$. By the definitions of $d_1, d_2$, and $d_3$, the distance matrix corresponding to $d_2$ is a homogenized version of $d_1$. This means that while retaining the main features of the distance matrix, $d_2$ smooths out $d_1$ everywhere, resulting in more evenly distributed entries throughout the matrix. 

\begin{figure}[!ht]
\begin{center}
\includegraphics[width=1.0\textwidth]{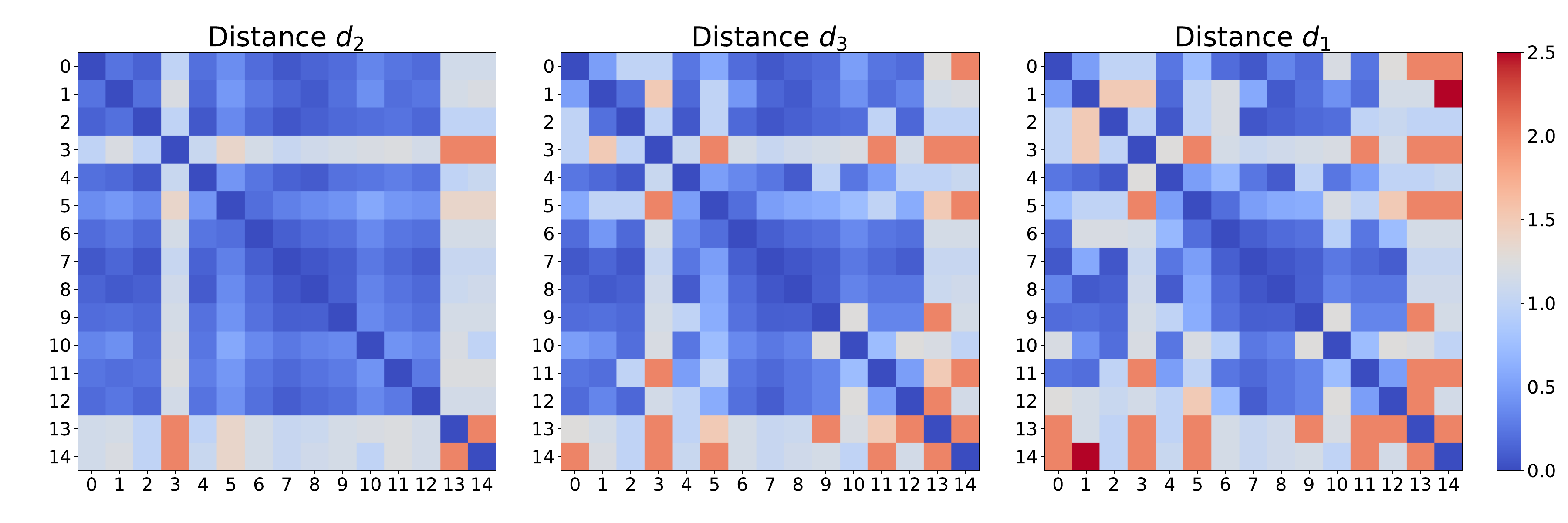}
\includegraphics[width=1.0\textwidth]{Figures/J-Sanghyeondodeuri_Geomungo_partdistmat_n.pdf}
\caption{Heatmap of distance matrices corresponding to $d_2$ (left), $d_3$ (middle) and $d_1$ (right). Top:  {\it Sangnyeongsan}.  Bottom: {\it Sanghyeondodeuri}. Both show the distance matrices played by the geomungo instrument. These figures clearly show that $d_2(v,w) \le d_3(v,w) \le d_1(v,w)$.}
\label{fig:heatmap}
\end{center}
\end{figure}

\begin{figure}[!ht]
\centering
    \includegraphics[width=4.8cm,height=4cm]{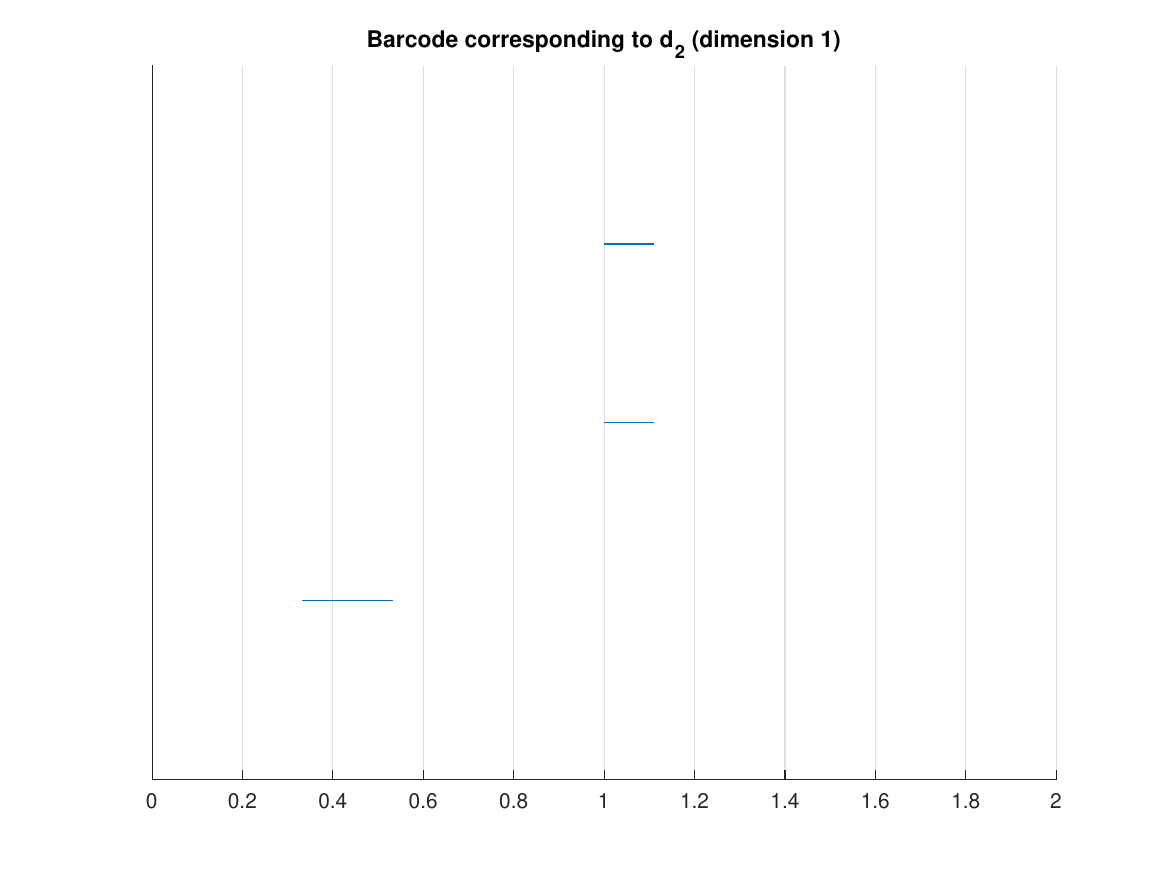}
    \includegraphics[width=4.8cm,height=4cm]{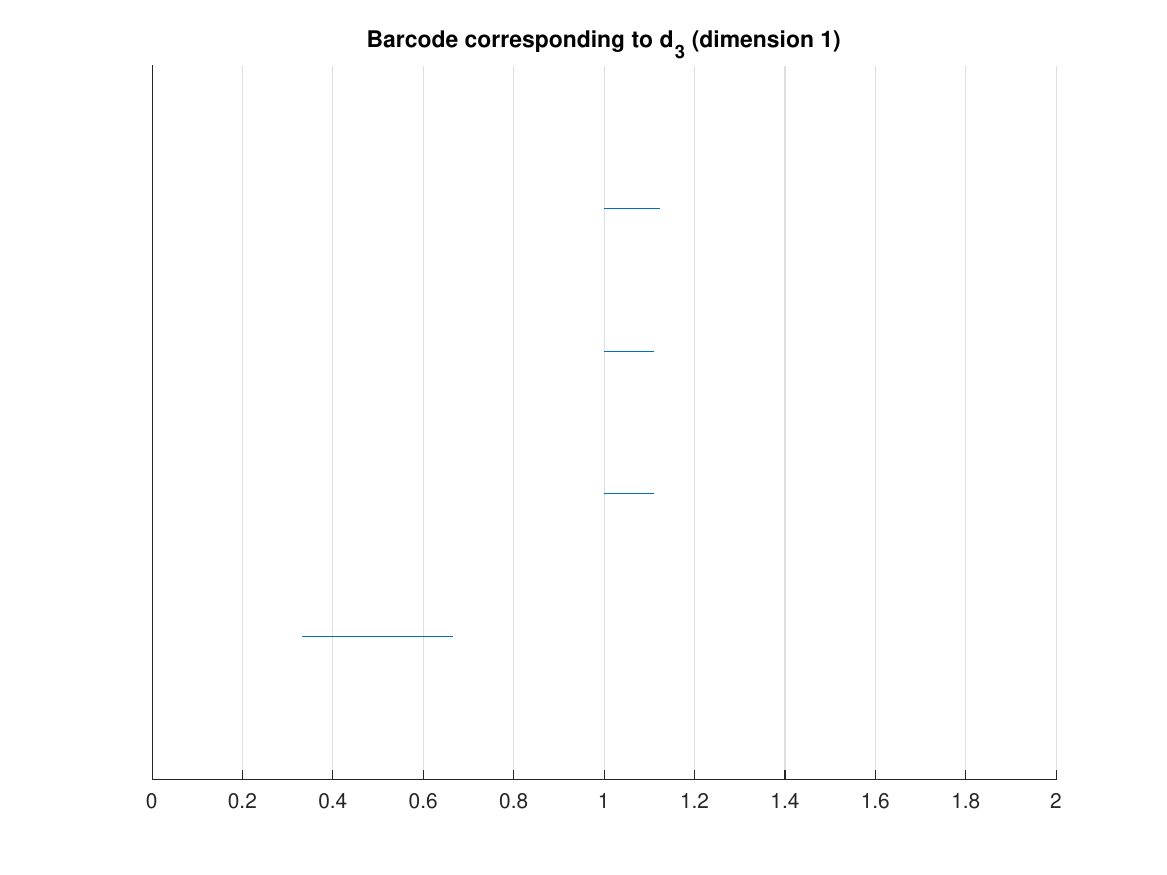}
    \includegraphics[width=4.8cm,height=4cm]{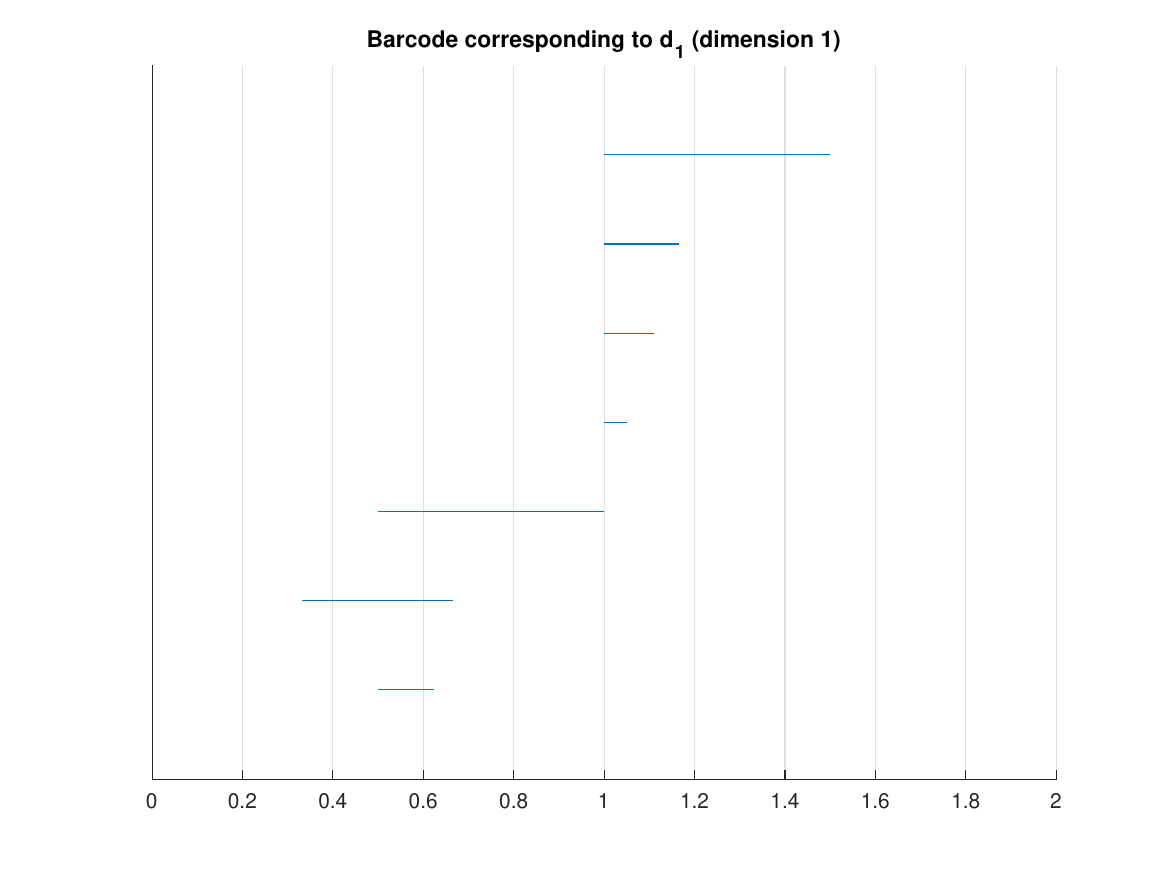}
    \includegraphics[width=10cm,height=8cm]{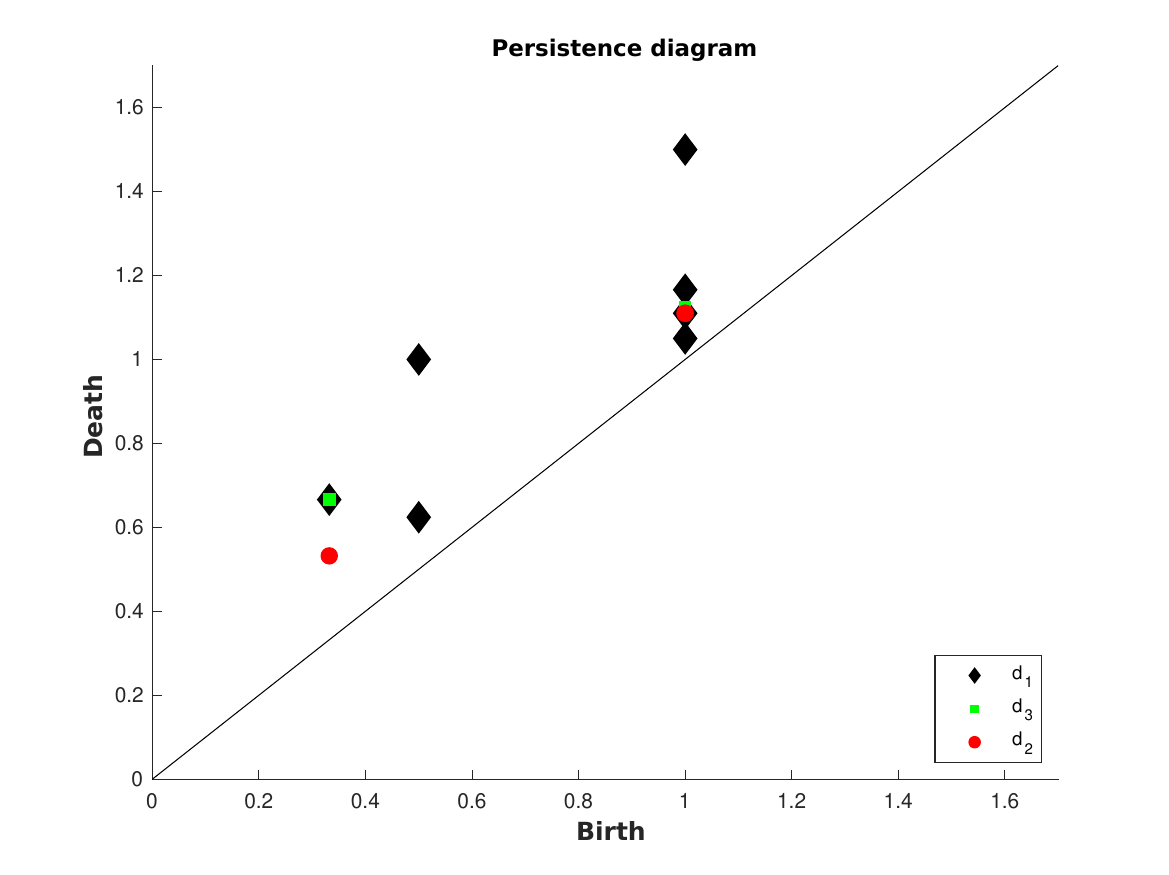}
\caption{Top: $H_1$ barcodes corresponding to $d_2$ (left), $d_3$ (middle), and $d_1$ (right). Bottom: $H_1$ persistence diagram. These results are for the music {\it Sanghyeondodeuri} (geomungo part). These figures clearly illustrate the inclusion relations among the persistence barcodes and show that the birth time is shared across all distance functions once a cycle is formed.}
\label{fig:bcNpd}
\end{figure}

The top figures of Figure~\ref{fig:bcNpd} show $H_1$ (one-dimensional homology) barcodes corresponding to $d_2$ (left), $d_3$ (middle), and $d_1$ (right) and the bottom figure shows $H_1$ persistence diagram that consists of the birth and death times of all persistence intervals in $d_2$ (red), $d_1$ (black), and $d_3$ (green) for {\it Sanghyeondodeuri} (geomungo part). 
The top figures show that there are seven $H_1$ barcodes for $d_1$, but only four $H_1$ barcodes for $d_3$ and three $H_1$ barcodes for $d_2$. Additionally, from the persistence diagram, we observe the three types of relationships between the birth and death times of all persistence intervals in $d_1, d_2$, and $d_3$, as demonstrated in the comprehensive persistence diagrams shown in Figure~\ref{fig:pdanalysis}. 
Note that this is a natural consequence of the injection relation $\mathcal{B}_2 \subseteq \mathcal{B}_3 \subseteq \mathcal{B}_1$. Also, the persistent diagram in Figure~\ref{fig:bcNpd} clearly shows that the birth time is shared across all distance functions once a cycle is formed.

The full comparison in terms of one-dimensional persistence barcodes and persistence diagrams by $d_1, d_2$, and $d_3$ for all the music samples considered is presented in Table~\ref{tab:compare}. We analyze twelve Korean songs and four variations of the simple Western song  {\it Twinkle Twinkle Little Star}. For Korean traditional music, we used  {\it Sangyeongsan}, {\it Jungnyeongsan}, {\it Seryeongsan}, {\it Garakdeori}, {\it Sanghyeondodeuri}, {\it Hahyeondodeuri}, {\it Yeombuldodeuri}, {\it Taryeong} and {\it Gunak}. For western music pieces, we used  {\it Twinkle Twinkle Little Star} and its variations. In the case of the Korean music pieces, either the geomungo or haegeum
\footnote{The haegeum is a traditional Korean string instrument, classified as a two-stringed vertical fiddle. It is played with a bow and held vertically on the performer's knee. Known for its distinctively expressive and somewhat nasal tone, the haegeum is highly versatile and used in various genres, from court and folk music to modern fusion styles.}
instrumental parts were used. The music titles in Table~\ref{tab:compare} that begin with ``J-" indicate that the pieces are part of {\it Junggwangjigok}.

For each song, the birth and death simplices, as well as the cycles with their corresponding persistence intervals, are provided for each distance $d_1, d_2$, and $d_3$. The cycles are sorted and enumerated by their birth times in a manner such that cycles that persist only in $d_1$ (or in both $d_1$ and $d_3$ but die in $d_2$) are placed later in the list. Note that the number of barcodes, i.e., the number of cycles, for $d_1$ is always greater than or equal to that for $d_3$, which is greater than or equal to that for $d_2$.  For each barcode that persists with both $d_1$ and $d_3$, or with all distances $d_1, d_2$, and $d_3$, the corresponding persistence interval with $d_1$ is longer than or equal to that with $d_3$, which is longer than or equal to that with $d_2$. Moreover, for each barcode, the birth simplices are the same for all the distances, while the death simplices are not necessarily the same. The cycles for each barcode can be, but are not always, the same for $d_1, d_2$, and $d_3$, even if the persistence intervals are the same, although the share a few common nodes. 
Tables~\ref{tab:x},~\ref{tab:y}, and~\ref{tab:z} present the same data as Table~\ref{tab:compare}, but are provided separately to enhance readability with larger font sizes.

\begin{table}[!ht]
\begin{center}
\resizebox{\textwidth}{!}{
\begin{tabular}{lclllllllll}
\hline\Tstrut\Bstrut
\multirow{2}{*}{\textbf{No}} & \multirow{2}{*}{\textbf{Song title}}& \multicolumn{3}{c} {$\mathbf{d_1} $} & \multicolumn{3}{c}{$\mathbf{d_3}$ } & \multicolumn{3}{c}{$\mathbf{d_2}$ } \\  
\cline{3-11}\Tstrut\Bstrut
 &   & \textbf{Birth}& \textbf{Death}& \textbf{Cycle}& \textbf{Birth}& \textbf{Death}& \textbf{Cycle}& \textbf{Birth}& \textbf{Death}& \textbf{Cycle}               \\
\hline\Tstrut
\multirow{2}{*}{1}  & \multirow{2}{*}{\shortstack{J-Sangnyeongsan \\ 
(geomungo part)} }   & (4, 8)   & (2, 7, 8)    & $C_1$ [1/11, 1/8]: [2, 4, 7, 8]  & (4, 8)   & (2, 7, 8)    & $C_1$ [1/11, 1/8]: [2, 4, 7, 8]  & (4, 8)   & (2, 7, 8)    & $C_1$ [1/11, 1/8]: [2, 4, 7, 8]  \\
   &                                                                                         & (1, 11)  & (6, 8, 11)   & $C_2$ [1/5, 1/4]: [0, 1, 6, 7, 8, 11]                                &          &              &                                                               &          &              &                                                               \Bstrut               \\
\hline \Tstrut 
\multirow{4}{*}{2}  & \multirow{4}{*}{\shortstack{J-Jungnyeongsan \\ (geomungo part)}}              & (4, 10)  & (2, 6, 10)   & $C_1$ [1/7, 5/28]: [2, 4, 6, 10]                                     & (4, 10)  & (2, 6, 10)   & $C_1$ [1/7, 5/28]: [2, 4, 6, 10]                         & (4, 10)  & (4, 6, 10)   & $C_1$ [1/7, 9/56]: [2, 4, 6, 10]                         \\
   &                                                                                         & (13, 14) & (2, 13, 14)  & $C_2$ [1, 6/5]: [2, 3, 6, 8, 13, 14]                                 & (13, 14) & (2, 13, 14)  & $C_2$ [1, 6/5]: [2, 3, 6, 8, 13, 14]                     & (13, 14) & (2, 13, 14)  & $C_2$ [1, 23/20]: [2, 3, 6, 8, 13, 14]                   \\
   &                                                                                         & (15, 16) & (11, 15, 16) & $C_3$ [1, 4/3]: [2, 4, 6, 10, 11, 15, 16]                            & (15, 16) & (10, 15, 16) & $C_3$ [1, 4/3]: [10, 11, 15, 16]                         & (15, 16) & (10, 15, 16) & $C_3$ [1, 4/3]: [10, 11, 15, 16]                         \\
   &                                                                                         & (19, 20) & (3, 18, 19)  & $C_4$ [1, 2]: [3, 17, 18, 19, 20]                                    & (19, 20) & (3, 18, 19)  & $C_4$ [1, 2]: [3, 17, 18, 19, 20]                        & (19, 20) & (3, 18, 19)  & $C_4$ [1, 2]: [3, 17, 18, 19, 20]                        \Bstrut               \\
\hline \Tstrut 
\multirow{3}{*}{3}  & \multirow{3}{*}{\shortstack{J-Seryeongsan \\ (geomungo part)}}                & (3, 4)   & (0, 3, 4)    & $C_1$ [1, 14/13]: [0, 2, 3, 4]                                       & (3, 4)   & (0, 3, 4)    & $C_1$ [1, 14/13]: [0, 2, 3, 4]                           & (3, 4)   & (0, 3, 4)    & $C_1$ [1, 14/13]: [0, 2, 3, 4]                           \\
   &                                                                                         & (17, 18) & (2, 17, 18)  & $C_2$ [1, 6/5]: [0, 2, 5, 17, 18]                                    & (17, 18) & (0, 17, 18)  & $C_2$ [1, 7/6]: [0, 5, 17, 18]                           & (17, 18) & (0, 17, 18)  & $C_2$ [1, 7/6]: [0, 5, 17, 18]                           \\
   &                                                                                         & (10, 11) & (2, 9, 10)   & $C_3$ [1, 2]: [0, 2, 5, 7, 8, 9, 10, 11]                             & (10, 11) & (2, 9, 10)   & $C_3$ [1, 2]: [0, 2, 5, 7, 8, 9, 10, 11]                 & (10, 11) & (2, 9, 10)   & $C_3$ [1, 2]: [0, 2, 5, 7, 8, 9, 10, 11]                \Bstrut               \\
\hline \Tstrut  
\multirow{5}{*}{4}  & \multirow{5}{*}{\shortstack{J-Garakdeori \\ (geomungo part)}}                 & (14, 17) & (0, 7, 17)   & $C_1$ [1/2, 1]: [0, 5, 6, 7, 14, 17]                                 & (14, 17) & (0, 7, 17)   & $C_1$ [1/2, 1]: [0, 5, 6, 7, 14, 17]                     & (14, 17) & (0, 7, 17)   & $C_1$ [1/2, 1]: [0, 5, 6, 7, 14, 17]                     \\
   &                                                                                         & (5, 6)   & (2, 6, 10)   & $C_2$ [1/2, 10/9]: [0, 2, 5, 6, 10, 11, 15, 16]                      & (5, 6)   & (5, 7, 14)   & $C_2$ [1/2, 1]: [0, 2, 5, 6, 10, 11, 14, 15, 16]         & (5, 6)   & (5, 7, 14)   & $C_2$ [1/2, 1]: [0, 2, 5, 6, 10, 11, 14, 15, 16]         \\
   &                                                                                         & (16, 19) & (11, 16, 19) & $C_3$ [1, 4/3]: [0, 2, 6, 11, 14, 16, 19]                            & (16, 19) & (10, 16, 19) & $C_3$ [1, 4/3]: [2, 6, 10, 14, 16, 19]                   & (16, 19) & (10, 16, 19) & $C_3$ [1, 4/3]: [2, 6, 10, 14, 16, 19]                   \\
   &                                                                                         & (3, 4)   & (0, 3, 4)    & $C_4$ [1, 2]: [0, 2, 3, 4]                                           & (3, 4)   & (3, 4, 5)    & $C_4$ [1, 29/18]: [0, 2, 3, 4]                           & (3, 4)   & (3, 4, 16)   & $C_4$ [1, 19/12]: [0, 2, 3, 4, 16]                       \\
   &                                                                                         & (9, 10)  & (2, 8, 9)    & $C_5$ [1, 2]: [0, 6, 7, 8, 9, 10]                                    & (9, 10)  & (8, 9, 17)   & $C_5$ [1, 11/6]: [0, 6, 7, 8, 9, 10, 17]                 & (9, 10)  & (8, 9, 17)   & $C_5$ [1, 11/6]: [0, 6, 7, 8, 9, 10]                      \Bstrut               \\
\hline \Tstrut 
\multirow{7}{*}{5}  & \multirow{7}{*}{\shortstack{J-Sanghyeondodeuri \\ (geomungo part)}}           & (12, 13) & (8, 12, 13)  & $C_1$ [1/3, 2/3]: [8, 11, 12, 13]                                    & (12, 13) & (8, 11, 12)  & $C_1$ [1/3, 2/3]: [1, 8, 11, 12, 13]                     & (12, 13) & (1, 12, 13)  & $C_1$ [1/3, 8/15]: [1, 8, 11, 12, 13]                    \\
   &                                                                                         & (33, 34) & (4, 33, 34)  & $C_2$ [1, 10/9]: [4, 5, 33, 34]                                      & (33, 34) & (4, 33, 34)  & $C_2$ [1, 10/9]: [4, 5, 33, 34]                          & (33, 34) & (4, 33, 34)  & $C_2$ [1, 10/9]: [4, 5, 33, 34]                          \\
   &                                                                                         & (24, 25) & (10, 24, 25) & $C_3$ [1, 7/6]: [1, 8, 9, 10, 24, 25]                                & (24, 25) & (8, 24, 25)  & $C_3$ [1, 10/9]: [8, 9, 24, 25]                          & (24, 25) & (8, 24, 25)  & $C_3$ [1, 10/9]: [8, 9, 24, 25]                          \\
   &                                                                                         & (2, 19)  & (2, 18, 19)  & $C_4$ [1, 3/2]: [0, 1, 2, 8, 9, 10, 17, 18, 19]                      & (2, 19)  & (1, 17, 19)  & $C_4$ [1, 9/8]: [1, 2, 17, 18, 19]                       &          &              &                                                               \\
   &                                                                                         & (8, 14)  & (4, 5, 14)   & $C_5$ [1/2, 5/8]: [1, 4, 5, 8, 14]                                   &          &              &                                                               &          &              &                                                               \\
   &                                                                                         & (3, 8)   & (3, 5, 8)    & $C_6$ [1/2, 1]: [3, 4, 5, 8]                                         &          &              &                                                               &          &              &                                                               \\
   &                                                                                         & (2, 3)   & (2, 3, 8)    & $C_7$ [1, 21/20]: [1, 2, 3, 8]                                       &          &              &                                                               &          &              &                                                               \Bstrut               \\
\hline \Tstrut 
\multirow{11}{*}{6}  & \multirow{11}{*}{\shortstack{J-Hahyeondodeuri \\ geomungo part)}}             & (14, 26) & (7, 14, 26)  & $C_1$ [1/2, 8/15]: [1, 4, 7, 14, 26]                                 & (14, 26) & (7, 14, 26)  & $C_1$ [1/2, 8/15]: [1, 4, 7, 14, 26]                     & (14, 26) & (7, 14, 26)  & $C_1$ [1/2, 8/15]: [1, 4, 7, 14, 26]                     \\
   &                                                                                         & (15, 25) & (15, 25, 33) & $C_2$ [1/2, 3/4]: [12, 15, 25, 33]                                   & (15, 25) & (12, 15, 33) & $C_2$ [1/2, 3/4]: [12, 15, 25, 33]                       & (15, 25) & (12, 15, 33) & $C_2$ [1/2, 3/4]: [12, 15, 25, 33]                       \\
   &                                                                                         & (12, 27) & (8, 10, 27)  & $C_3$ [1, 7/6]: [4, 8, 10, 12, 27]                                   & (12, 27) & (8, 10, 27)  & $C_3$ [1, 7/6]: [4, 8, 10, 12, 27]                       & (12, 27) & (8, 10, 27)  & $C_3$ [1, 7/6]: [4, 8, 10, 12, 27]                       \\
   &                                                                                         & (6, 7)   & (6, 7, 8)    & $C_4$ [1, 4/3]: [4, 5, 6, 7, 8]                                      & (6, 7)   & (4, 6, 7)    & $C_4$ [1, 6/5]: [4, 5, 6, 7]                             & (6, 7)   & (4, 6, 7)    & $C_4$ [1, 6/5]: [4, 5, 6, 7]                             \\
   &                                                                                         & (15, 36) & (4, 33, 36)  & $C_5$ [1, 3/2]: [0, 2, 4, 15, 36]                                    & (15, 36) & (1, 15, 36)  & $C_5$ [1, 3/2]: [0, 1, 2, 15, 36]                        & (15, 36) & (0, 33, 36)  & $C_5$ [1, 83/66]: [0, 1, 15, 33, 36]                     \\
   &                                                                                         & (29, 30) & (12, 29, 30) & $C_6$ [1, 3/2]: [11, 12, 25, 29, 30]                                 & (29, 30) & (12, 29, 30) & $C_6$ [1, 3/2]: [11, 12, 25, 29, 30]                     & (29, 30) & (12, 29, 30) & $C_6$ [1, 3/2]: [11, 12, 25, 29, 30]                     \\
   &                                                                                         & (12, 19) & (4, 12, 19)  & $C_7$ [1/2, 1]: [4, 8, 12, 19]                                       & (12, 19) & (12, 19, 25) & $C_7$ [1/2, 3/4]: [4, 12, 19, 25]                        &          &              &                                                               \\
   &                                                                                         & (2, 15)  & (7, 8, 15)   & $C_8$ [1, 4/3]: [1, 2, 7, 8, 15, 25]                                 & (2, 15)  & (2, 8, 15)   & $C_8$ [1, 7/6]: [1, 2, 8, 12, 15]                        &          &              &                                                               \\
   &                                                                                         & (15, 32) & (8, 15, 32)  & $C_9$ [1, 12/11]: [8, 12, 15, 32]                                    &          &              &                                                               &          &              &                                                               \\
   &                                                                                         & (10, 11) & (8, 10, 11)  & $C_{10}$ [1, 7/6]: [4, 8, 10, 11]                                      &          &              &                                                               &          &              &                                                               \\
   &                                                                                         & (16, 17) & (8, 16, 17)  & $C_{11}$ [1, 7/6]: [4, 8, 16, 17]                                      &          &              &                                                               &          &              &                                                               \Bstrut               \\
\hline \Tstrut 
\multirow{6}{*}{7}  & \multirow{6}{*}{\shortstack{J-Yeombuldodeuri \\ (geomungo part)}}             & (5, 6)   & (4, 6, 7)    & $C_1$ [1/5, 1/3]: [2, 3, 4, 5, 6, 7]                                 & (5, 6)   & (2, 5, 7)    & $C_1$ [1/5, 7/24]: [2, 5, 6, 7]                          & (5, 6)   & (4, 5, 6)    & $C_1$ [1/5, 1/4]: [2, 3, 4, 5, 6]                        \\
   &                                                                                         & (28, 29) & (3, 28, 29)  & $C_2$ [1, 5/4]: [0, 3, 25, 28, 29]                                   & (28, 29) & (3, 28, 29)  & $C_2$ [1, 5/4]: [0, 3, 25, 28, 29]                       & (28, 29) & (3, 28, 29)  & $C_2$ [1, 5/4]: [0, 3, 25, 28, 29]                       \\
   &                                                                                         & (5, 7)   & (4, 5, 7)    & $C_3$ [1/8, 1/4]: [3, 4, 5, 7]                                       & (5, 7)   & (4, 5, 7)    & $C_3$ [1/8, 1/4]: [3, 4, 5, 7]                           &          &              &                                                               \\
   &                                                                                         & (5, 10)  & (4, 10, 11)  & $C_4$ [1/2, 7/12]: [0, 4, 5, 10, 11]                                 & (5, 10)  & (4, 10, 11)  & $C_4$ [1/2, 11/20]: [4, 5, 10, 11]                       &          &              &                                                               \\
   &                                                                                         & (3, 20)  & (6, 19, 20)  & $C_5$ [1/4, 1/2]: [2, 3, 4, 5, 6, 19, 20]                            &          &              &                                                               &          &              &                                                               \\
   &                                                                                         & (18, 19) & (4, 18, 19)  & $C_6$ [1, 26/25]: [3, 4, 18, 19]                                     &          &              &                                                               &          &              &                                                              \Bstrut               \\
\hline \Tstrut 
\multirow{6}{*}{8}  & \multirow{6}{*}{\shortstack{J-Taryeong \\ (geomungo part)}}                   & (2, 12)  & (2, 11, 12)  & $C_1$ [1/3, 10/21]: [1, 2, 10, 11, 12]                               & (2, 12)  & (2, 4, 11)   & $C_1$ [1/3, 10/21]: [1, 2, 4, 10, 11, 12]                & (2, 12)  & (2, 10, 12)  & $C_1$ [1/3, 389/1012]: [1, 2, 10, 12]                    \\
   &                                                                                         & (4, 7)   & (6, 10, 12)  & $C_2$ [1/3, 1/2]: [1, 2, 4, 6, 7, 10, 12]                            & (4, 7)   & (2, 4, 12)   & $C_2$ [1/3, 22/57]: [2, 4, 6, 7, 12]                     & (4, 7)   & (4, 10, 11)  & $C_2$ [1/3, 45/133]: [4, 7, 10, 11]                      \\
   &                                                                                         & (32, 33) & (6, 32, 33)  & $C_3$ [1, 21/20]: [1, 6, 7, 32, 33]                                  & (32, 33) & (6, 32, 33)  & $C_3$ [1, 21/20]: [1, 6, 7, 32, 33]                      & (32, 33) & (6, 32, 33)  & $C_3$ [1, 21/20]: [1, 6, 7, 32, 33]                      \\
   &                                                                                         & (27, 28) & (11, 27, 28) & $C_4$ [1, 47/35]: [0, 1, 3, 4, 11, 27, 28]                           & (27, 28) & (7, 27, 28)  & $C_4$ [1, 4/3]: [1, 3, 4, 7, 27, 28]                     & (27, 28) & (7, 27, 28)  & $C_4$ [1, 4/3]: [1, 3, 4, 7, 27, 28]                     \\
   &                                                                                         & (1, 2)   & (1, 6, 14)   & $C_5$ [1/4, 15/44]: [1, 2, 6, 14]                                    &          &              &                                                               &          &              &                                                               \\
   &                                                                                         & (22, 24) & (1, 22, 24)  & $C_6$ [1, 59/44]: [0, 1, 2, 6, 14, 22, 24]                           &          &              &                                                               &          &              &                                                              \Bstrut               \\
\hline \Tstrut 
\multirow{12}{*}{9}  & \multirow{12}{*}{\shortstack{J-Gunak \\ (geomungo part)}}                      & (0, 1)   & (1, 5, 6)    & $C_1$ [1/6, 19/90]: [0, 1, 5, 6]                                     & (0, 1)   & (1, 4, 6)    & $C_1$ [1/6, 19/90]: [0, 1, 5, 6]                         & (0, 1)   & (1, 4, 6)    & $C_1$ [1/6, 19/90]: [0, 1, 5, 6]                         \\
   &                                                                                         & (10, 11) & (4, 10, 11)  & $C_2$ [1/8, 19/88]: [3, 4, 10, 11]                                   & (10, 11) & (3, 10, 11)  & $C_2$ [1/8, 19/88]: [3, 4, 10, 11]                       & (10, 11) & (3, 10, 11)  & $C_2$ [1/8, 19/88]: [3, 4, 10, 11]                       \\
   &                                                                                         & (1, 2)   & (2, 10, 11)  & $C_3$ [1/7, 1/4]: [1, 2, 4, 5, 10, 11]                               & (1, 2)   & (2, 3, 5)    & $C_3$ [1/7, 17/70]: [1, 2, 3, 4, 5, 11]                  & (1, 2)   & (2, 3, 5)    & $C_3$ [1/7, 17/70]: [1, 2, 3, 4, 5, 11]                  \\
   &                                                                                         & (11, 16) & (3, 11, 16)  & $C_4$ [1/3, 14/33]: [0, 1, 3, 11, 16]                                & (11, 16) & (3, 11, 16)  & $C_4$ [1/3, 14/33]: [0, 1, 3, 11, 16]                    & (11, 16) & (6, 11, 16)  & $C_4$ [1/3, 2209/5544]: [0, 1, 6, 11, 16]                \\
   &                                                                                         & (16, 24) & (11, 16, 24) & $C_5$ [1/3, 2/3]: [4, 11, 16, 17, 24]                                & (16, 24) & (6, 9, 16)   & $C_5$ [1/3, 11/18]: [0, 1, 4, 5, 6, 9, 16, 17, 24]       & (16, 24) & (9, 16, 24)  & $C_5$ [1/3, 8/15]: [0, 4, 5, 9, 16, 17, 24]              \\
   &                                                                                         & (19, 31) & (11, 19, 31) & $C_6$ [1, 71/55]: [3, 7, 8, 19, 31]                                  & (19, 31) & (11, 19, 31) & $C_6$ [1, 71/55]: [3, 4, 8, 11, 19, 31]                  & (19, 31) & (8, 19, 31)  & $C_6$ [1, 139/120]: [0, 1, 3, 8, 19, 31]                 \\
   &                                                                                         & (32, 33) & (29, 32, 33) & $C_7$ [1, 23/14]: [17, 19, 29, 30, 32, 33]                           & (32, 33) & (17, 32, 33) & $C_7$ [1, 3/2]: [17, 19, 29, 30, 32, 33]                 & (32, 33) & (17, 32, 33) & $C_7$ [1, 3/2]: [17, 19, 29, 30, 32, 33]                 \\
   &                                                                                         & (25, 26) & (24, 25, 26) & $C_8$ [1, 17/10]: [0, 9, 16, 18, 24, 25, 26]                         & (25, 26) & (9, 25, 26)  & $C_8$ [1, 23/15]: [0, 1, 9, 16, 18, 25, 26]              & (25, 26) & (9, 25, 26)  & $C_8$ [1, 23/15]: [0, 9, 16, 18, 25, 26]                 \\
   &                                                                                         & (7, 19)  & (9, 11, 19)  & $C_9$ [1, 11/10]: [2, 3, 7, 9, 11, 19]                               & (6, 7)   & (6, 7, 11)   & $C_9$ [1, 12/11]: [3, 6, 7, 11]                          &          &              &                                                               \\
   &                                                                                         & (6, 7)   & (5, 6, 7)    & $C_{10}$ [1, 10/9]: [3, 5, 6, 7]                                       & (7, 19)  & (7, 11, 19)  & $C_{10}$ [1, 11/10]: [3, 4, 7, 11, 19]                     &          &              &                                                               \\
   &                                                                                         & (10, 15) & (5, 10, 15)  & $C_{11}$ [1/6, 11/56]: [4, 5, 10, 15]                                  &          &              &                                                               &          &              &                                                               \\
   &                                                                                         & (2, 8)   & (2, 3, 8)    & $C_{12}$ [1/3, 11/28]: [1, 2, 3, 8]                                    &          &              &                                                               &          &              &                                                              \Bstrut               \\
\hline \Tstrut 
\multirow{10}{*}{10} & \multirow{10}{*}{\shortstack{J-Sanghyeondodeuri \\ (haegeum part)}} & (1, 17)  & (4, 10, 17)  & $C_1$ [1/4, 1/3]: [1, 4, 10, 17]                                     & (1, 17)  & (4, 6, 17)   & $C_1$ [1/4, 23/76]: [1, 4, 6, 10, 17]                    & (1, 17)  & (4, 6, 17)   & $C_1$ [1/4, 23/76]: [1, 4, 6, 10, 17]                    \\
   &                                                                                         & (1, 11)  & (1, 11, 13)  & $C_2$ [1/5, 1/2]: [1, 4, 10, 11, 13]                                 & (1, 11)  & (4, 11, 13)  & $C_2$ [1/5, 13/42]: [1, 4, 6, 10, 11, 13]                & (1, 11)  & (4, 11, 13)  & $C_2$ [1/5, 13/42]: [1, 4, 6, 10, 11, 13]                \\
   &                                                                                         & (3, 7)   & (3, 4, 7)    & $C_3$ [1/3, 1/2]: [1, 3, 4, 7, 10, 11, 12]                           & (3, 7)   & (3, 4, 7)    & $C_3$ [1/3, 1/2]: [1, 3, 4, 6, 7, 11, 12]                & (3, 7)   & (1, 7, 12)   & $C_3$ [1/3, 1/2]: [1, 3, 4, 7, 11, 12]                   \\
   &                                                                                         & (22, 23) & (1, 22, 23)  & $C_4$ [1, 6/5]: [1, 4, 10, 21, 22, 23]                               & (22, 23) & (1, 22, 23)  & $C_4$ [1, 6/5]: [1, 4, 6, 10, 21, 22, 23]                & (22, 23) & (1, 22, 23)  & $C_4$ [1, 6/5]: [1, 6, 10, 21, 22, 23]                   \\
   &                                                                                         & (16, 27) & (2, 16, 27)  & $C_5$ [1, 5/4]: [1, 2, 16, 27]                                       & (16, 27) & (1, 16, 27)  & $C_5$ [1, 5/4]: [1, 2, 16, 27]                           & (16, 27) & (1, 16, 27)  & $C_5$ [1, 5/4]: [1, 2, 16, 27]                           \\
   &                                                                                         & (8, 9)   & (1, 8, 9)    & $C_6$ [1, 3/2]: [1, 4, 6, 7, 8, 9, 10]                               & (8, 9)   & (8, 9, 12)   & $C_6$ [1, 29/20]: [1, 6, 7, 8, 9, 12]                    & (8, 9)   & (3, 8, 9)    & $C_6$ [1, 4/3]: [1, 3, 6, 7, 8, 9]                       \\
   &                                                                                         & (4, 15)  & (6, 10, 15)  & $C_7$ [1/3, 1/2]: [4, 6, 10, 15]                                     &          &              &                                                               &          &              &                                                               \\
   &                                                                                         & (0, 15)  & (0, 10, 15)  & $C_8$ [1/2, 8/15]: [0, 1, 4, 10, 15]                                 &          &              &                                                               &          &              &                                                               \\
   &                                                                                         & (15, 25) & (6, 15, 25)  & $C_9$ [1/2, 7/12]: [6, 13, 15, 25]                                   &          &              &                                                               &          &              &                                                               \\
   &                                                                                         & (5, 6)   & (5, 6, 13)   & $C_{10}$ [1, 8/7]: [4, 5, 6, 13]                                       &          &              &                                                               &          &              &                                                               \Bstrut               \\
\hline \Tstrut 
\multirow{21}{*}{11} & \multirow{21}{*}{\shortstack{J-Hahyeondodeuri \\ (haegeum part)}}   & (20, 21) & (20, 22, 23) & $C_1$ [1/2, 7/12]: [20, 21, 22, 23]                                  & (20, 21) & (20, 22, 23) & $C_1$ [1/2, 7/12]: [20, 21, 22, 23]                      & (20, 21) & (7, 20, 21)  & $C_1$ [1/2, 7/12]: [6, 7, 20, 21]                        \\
   &                                                                                         & (18, 21) & (12, 18, 21) & $C_2$ [1/2, 3/4]: [7, 12, 18, 21]                                    & (18, 21) & (6, 18, 21)  & $C_2$ [1/2, 13/22]: [6, 7, 18, 21]                       & (18, 21) & (7, 18, 21)  & $C_2$ [1/2, 7/12]: [6, 7, 18, 21]                        \\
   &                                                                                         & (7, 20)  & (6, 20, 21)  & $C_3$ [1/2, 13/22]: [6, 7, 20, 21, 22, 23]                           & (7, 20)  & (6, 20, 21)  & $C_3$ [1/2, 13/22]: [6, 7, 20, 21, 22, 23]               & (7, 20)  & (20, 22, 23) & $C_3$ [1/2, 7/12]: [6, 7, 20, 21, 22, 23]                \\
   &                                                                                         & (17, 23) & (10, 17, 23) & $C_4$ [1/2, 3/4]: [7, 8, 10, 17, 20, 23]                             & (17, 23) & (7, 20, 23)  & $C_4$ [1/2, 3/4]: [7, 8, 9, 17, 20, 23]                  & (17, 23) & (7, 20, 23)  & $C_4$ [1/2, 3/4]: [7, 8, 9, 17, 20, 23]                  \\
   &                                                                                         & (15, 16) & (13, 15, 21) & $C_5$ [1/2, 31/30]: [4, 7, 8, 9, 10, 13, 14, 15, 16, 20, 21, 22, 23] & (15, 16) & (13, 15, 21) & $C_5$ [1/2, 31/30]: [6, 7, 8, 9, 10, 13, 14, 15, 16]     & (15, 16) & (13, 15, 21) & $C_5$ [1/2, 31/30]: [6, 7, 8, 9, 10, 13, 14, 15, 16]     \\
   &                                                                                         & (24, 26) & (4, 24, 26)  & $C_6$ [1, 3/2]: [2, 3, 4, 24, 26]                                    & (24, 26) & (4, 24, 26)  & $C_6$ [1, 3/2]: [2, 3, 4, 24, 26]                        & (24, 26) & (3, 24, 26)  & $C_6$ [1, 4/3]: [3, 4, 24, 26]                           \\
   &                                                                                         & (24, 27) & (4, 24, 27)  & $C_7$ [1, 3/2]: [2, 3, 4, 24, 27]                                    & (24, 27) & (4, 24, 27)  & $C_7$ [1, 3/2]: [2, 3, 4, 24, 27]                        & (24, 27) & (3, 24, 27)  & $C_7$ [1, 4/3]: [3, 4, 24, 27]                           \\
   &                                                                                         & (32, 33) & (21, 32, 33) & $C_8$ [1, 4/3]: [4, 7, 21, 32, 33]                                   & (32, 33) & (21, 32, 33) & $C_8$ [1, 4/3]: [4, 7, 21, 32, 33]                       & (32, 33) & (21, 32, 33) & $C_8$ [1, 4/3]: [4, 6, 7, 21, 32, 33]                    \\
   &                                                                                         & (10, 30) & (1, 9, 30)   & $C_9$ [1, 3/2]: [1, 2, 9, 10, 30]                                    & (10, 30) & (1, 9, 30)   & $C_9$ [1, 3/2]: [1, 2, 9, 10, 30]                        & (10, 30) & (1, 9, 30)   & $C_9$ [1, 3/2]: [0, 1, 2, 9, 10, 30]                     \\
   &                                                                                         & (28, 29) & (12, 28, 29) & $C_{10}$ [1, 5/3]: [0, 2, 8, 9, 12, 28, 29]                            & (28, 29) & (12, 28, 29) & $C_{10}$ [1, 5/3]: [0, 2, 7, 8, 9, 12, 28, 29]             & (28, 29) & (12, 28, 29) & $C_{10}$ [1, 268/165]: [0, 2, 7, 8, 9, 12, 28, 29]         \\
   &                                                                                         & (6, 7)   & (6, 12, 21)  & $C_{11}$ [1/2, 1]: [6, 7, 12, 21]                                      & (6, 7)   & (10, 18, 21) & $C_{11}$ [1/2, 7/10]: [6, 7, 8, 9, 10, 12, 18, 21]         &          &              &                                                               \\
   &                                                                                         & (20, 35) & (6, 20, 35)  & $C_{12}$ [1/2, 13/22]: [6, 7, 20, 35]                                  &          &              &                                                               &          &              &                                                               \\
   &                                                                                         & (14, 24) & (14, 18, 24) & $C_{13}$ [1, 31/30]: [6, 9, 10, 14, 18, 24]                            &          &              &                                                               &          &              &                                                               \\
   &                                                                                         & (5, 6)   & (5, 6, 21)   & $C_{14}$ [1, 12/11]: [4, 5, 6, 21]                                     &          &              &                                                               &          &              &                                                               \\
   &                                                                                         & (6, 19)  & (6, 19, 21)  & $C_{15}$ [1, 12/11]: [4, 6, 19, 21]                                    &          &              &                                                               &          &              &                                                               \\
   &                                                                                         & (0, 13)  & (2, 8, 20)   & $C_{16}$ [1, 8/7]: [0, 4, 8, 9, 13, 17, 20]                            &          &              &                                                               &          &              &                                                               \\
   &                                                                                         & (16, 17) & (8, 16, 17)  & $C_{17}$ [1, 8/7]: [2, 4, 8, 9, 16, 17]                                &          &              &                                                               &          &              &                                                               \\
   &                                                                                         & (22, 33) & (18, 22, 33) & $C_{18}$ [1, 8/7]: [7, 18, 22, 33]                                     &          &              &                                                               &          &              &                                                               \\
   &                                                                                         & (11, 12) & (9, 11, 12)  & $C_{19}$ [1, 6/5]: [9, 10, 11, 12]                                     &          &              &                                                               &          &              &                                                               \\
   &                                                                                         & (28, 31) & (6, 21, 28)  & $C_{20}$ [1, 4/3]: [4, 6, 21, 28, 31]                                  &          &              &                                                               &          &              &                                                               \\
   &                                                                                         & (6, 34)  & (6, 21, 34)  & $C_{21}$ [1, 4/3]: [4, 6, 21, 34]                                      &          &              &                                                               &          &              &                                                               \Bstrut               \\
\hline \Tstrut 
\multirow{19}{*}{12} & \multirow{19}{*}{\shortstack{J-Yeombuldodeuri \\ (haegeum part)}}   & (0, 9)   & (8, 9, 16)   & $C_1$ [1/4, 1/2]: [0, 3, 9, 10, 14, 16]                              & (0, 9)   & (9, 10, 16)  & $C_1$ [1/4, 18/65]: [0, 9, 10, 14, 16]                   & (0, 9)   & (8, 10, 16)  & $C_1$ [1/4, 18/65]: [0, 3, 9, 10, 16]                    \\
   &                                                                                         & (5, 6)   & (2, 3, 5)    & $C_2$ [1/3, 1/2]: [0, 2, 3, 4, 5, 6]                                 & (5, 6)   & (2, 5, 18)   & $C_2$ [1/3, 13/30]: [2, 5, 6, 18]                        & (5, 6)   & (5, 6, 18)   & $C_2$ [1/3, 23/60]: [0, 5, 6, 18]                        \\
   &                                                                                         & (5, 18)  & (4, 5, 18)   & $C_3$ [1/5, 2/5]: [0, 3, 4, 5, 18]                                   & (5, 18)  & (3, 5, 18)   & $C_3$ [1/5, 2/5]: [0, 3, 4, 5, 18]                       & (5, 18)  & (0, 5, 6)    & $C_3$ [1/5, 2/5]: [0, 3, 4, 5, 18]                       \\
   &                                                                                         & (5, 10)  & (4, 10, 16)  & $C_4$ [1/3, 8/15]: [0, 4, 5, 10, 16]                                 & (5, 10)  & (3, 5, 10)   & $C_4$ [1/3, 2/5]: [3, 4, 5, 10]                          & (5, 10)  & (0, 5, 10)   & $C_4$ [1/3, 2/5]: [0, 4, 5, 10]                          \\
   &                                                                                         & (11, 13) & (10, 11, 13) & $C_5$ [1/2, 7/10]: [5, 9, 10, 11, 13]                                & (11, 13) & (11, 13, 18) & $C_5$ [1/2, 8/15]: [4, 5, 11, 13, 18]                    & (11, 13) & (11, 13, 18) & $C_5$ [1/2, 8/15]: [4, 5, 11, 13, 18]                    \\
   &                                                                                         & (19, 20) & (0, 19, 20)  & $C_6$ [1/2, 7/10]: [0, 18, 19, 20]                                   & (19, 20) & (0, 19, 20)  & $C_6$ [1/2, 7/10]: [0, 18, 19, 20]                       & (19, 20) & (18, 19, 20) & $C_6$ [1/2, 8/15]: [0, 18, 19, 20]                       \\
   &                                                                                         & (1, 12)  & (1, 9, 12)   & $C_7$ [1/2, 1]: [0, 1, 2, 9, 11, 12]                                 & (1, 12)  & (1, 9, 12)   & $C_7$ [1/2, 5/6]: [0, 1, 9, 11, 12]                      & (1, 12)  & (1, 2, 12)   & $C_7$ [1/2, 7/10]: [0, 1, 2, 9, 11, 12]                  \\
   &                                                                                         & (34, 35) & (9, 34, 35)  & $C_8$ [1/2, 5/6]: [9, 12, 34, 35]                                    & (34, 35) & (12, 34, 35) & $C_8$ [1/2, 7/10]: [0, 9, 12, 34, 35]                    & (34, 35) & (12, 34, 35) & $C_8$ [1/2, 7/10]: [0, 9, 12, 34, 35]                    \\
   &                                                                                         & (19, 39) & (5, 19, 39)  & $C_9$ [1, 4/3]: [0, 5, 8, 10, 19, 39]                                & (19, 39) & (14, 19, 39) & $C_9$ [1, 153/130]: [10, 14, 19, 39]                     & (19, 39) & (10, 19, 39) & $C_9$ [1, 267/260]: [0, 10, 19, 39]                      \\
   &                                                                                         & (37, 40) & (5, 37, 40)  & $C_{10}$ [1, 4/3]: [0, 5, 8, 10, 18, 37, 40]                           & (37, 40) & (5, 37, 40)  & $C_{10}$ [1, 4/3]: [0, 3, 5, 8, 10, 18, 37, 40]            & (37, 40) & (3, 37, 40)  & $C_{10}$ [1, 677/520]: [0, 3, 8, 10, 18, 37, 40]           \\
   &                                                                                         & (30, 31) & (0, 30, 31)  & $C_{11}$ [1, 3/2]: [0, 3, 10, 21, 24, 30, 31]                          & (30, 31) & (0, 30, 31)  & $C_{11}$ [1, 3/2]: [0, 3, 10, 21, 24, 30, 31]              & (30, 31) & (16, 30, 31) & $C_{11}$ [1, 59/40]: [0, 3, 10, 16, 21, 24, 30, 31]        \\
   &                                                                                         & (35, 41) & (0, 35, 41)  & $C_{12}$ [1, 6/5]: [0, 18, 35, 41]                                     & (35, 41) & (0, 35, 41)  & $C_{12}$ [1, 6/5]: [0, 18, 35, 41]                         &          &              &                                                               \\
   &                                                                                         & (25, 26) & (18, 25, 26) & $C_{13}$ [1, 6/5]: [0, 18, 25, 26]                                     & (25, 26) & (18, 25, 26) & $C_{13}$ [1, 6/5]: [0, 18, 25, 26]                         &          &              &                                                               \\
   &                                                                                         & (2, 14)  & (3, 14, 22)  & $C_{14}$ [1/3, 3/8]: [0, 2, 3, 14, 22]                                 &          &              &                                                               &          &              &                                                               \\
   &                                                                                         & (0, 24)  & (9, 10, 24)  & $C_{15}$ [1/2, 8/15]: [0, 9, 10, 24]                                   &          &              &                                                               &          &              &                                                               \\
   &                                                                                         & (9, 11)  & (3, 9, 11)   & $C_{16}$ [1/2, 7/10]: [2, 3, 9, 11]                                    &          &              &                                                               &          &              &                                                               \\
   &                                                                                         & (8, 28)  & (10, 28, 29) & $C_{17}$ [1, 14/13]: [8, 10, 28, 29]                                   &          &              &                                                               &          &              &                                                               \\
   &                                                                                         & (24, 40) & (5, 29, 40)  & $C_{18}$ [1, 6/5]: [5, 18, 24, 29, 40]                                 &          &              &                                                               &          &              &                                                               \\
   &                                                                                         & (24, 36) & (24, 33, 36) & $C_{19}$ [1, 6/5]: [14, 16, 24, 33, 36]                                &          &              &                                                               &          &              &                                                                \Bstrut               \\
\hline \Tstrut 
\multirow{4}{*}{13} & \multirow{4}{*}{\shortstack{Twinkle Twinkle Little Star \\ (theme)}}          & (8, 9)   & (5, 8, 9)    & $C_1$ [1/2, 5/8]: [3, 4, 5, 8, 9]                                    & (8, 9)   & (4, 8, 9)    & $C_1$ [1/2, 5/8]: [4, 5, 8, 9]                           & (8, 9)   & (4, 8, 9)    & $C_1$ [1/2, 5/8]: [4, 5, 8, 9]                           \\
   &                                                                                         & (5, 6)   & (5, 6, 7)    & $C_2$ [1/2, 3/4]: [0, 5, 6, 7]                                       & (5, 6)   & (0, 6, 7)    & $C_2$ [1/2, 3/4]: [0, 1, 4, 5, 6, 7]                     & (5, 6)   & (0, 6, 7)    & $C_2$ [1/2, 3/4]: [0, 4, 5, 6, 7]                        \\
   &                                                                                         & (11, 12) & (4, 11, 12)  & $C_3$ [1/2, 1]: [0, 1, 4, 5, 9, 10, 11, 12]                          & (11, 12) & (4, 11, 12)  & $C_3$ [1/2, 1]: [4, 5, 9, 10, 11, 12]                    & (11, 12) & (4, 11, 12)  & $C_3$ [1/2, 1]: [4, 5, 9, 10, 11, 12]                    \\
   &                                                                                         & (10, 13) & (5, 9, 13)   & $C_4$ [1/2, 1]: [5, 6, 9, 10, 13]                                    & (10, 13) & (5, 9, 13)   & $C_4$ [1/2, 1]: [5, 6, 9, 10, 13]                        & (10, 13) & (5, 9, 13)   & $C_4$ [1/2, 1]: [5, 6, 9, 10, 13]                        \Bstrut               \\
\hline \Tstrut
\multirow{6}{*}{14} & \multirow{6}{*}{\shortstack{Twinkle Twinkle Little Star \\ (variation 1)}}    & (1, 3)   & (1, 3, 14)   & $C_1$ [1/4, 5/12]: [0, 1, 3, 7, 14]                                  & (1, 3)   & (1, 3, 12)   & $C_1$ [1/4, 73/208]: [0, 1, 3, 12]                       & (1, 3)   & (1, 3, 4)    & $C_1$ [1/4, 15/52]: [0, 1, 3, 4, 12]                     \\
   &                                                                                         & (15, 16) & (11, 15, 16) & $C_2$ [1/4, 1/2]: [3, 7, 11, 12, 14, 15, 16]                         & (15, 16) & (11, 15, 16) & $C_2$ [1/4, 1/2]: [3, 4, 11, 12, 14, 15, 16]             & (15, 16) & (3, 15, 16)  & $C_2$ [1/4, 1/2]: [3, 4, 11, 12, 14, 15, 16]             \\
   &                                                                                         & (9, 12)  & (7, 12, 14)  & $C_3$ [1/6, 1/4]: [7, 9, 12, 14]                                     & (9, 12)  & (7, 9, 12)   & $C_3$ [1/6, 8/39]: [3, 7, 9, 12]                         &          &              &                                                               \\
   &                                                                                         & (7, 14)  & (8, 11, 14)  & $C_4$ [1/6, 5/16]: [3, 4, 7, 8, 12, 14]                              & (7, 14)  & (7, 12, 14)  & $C_4$ [1/6, 8/39]: [3, 7, 12, 14]                        &          &              &                                                               \\
   &                                                                                         & (11, 12) & (4, 11, 12)  & $C_5$ [1/4, 15/52]: [3, 4, 11, 12]                                   &          &              &                                                               &          &              &                                                               \\
   &                                                                                         & (13, 14) & (12, 13, 14) & $C_6$ [1/4, 5/16]: [4, 12, 13, 14]                                   &          &              &                                                               &          &              &                                                              \Bstrut               \\
\hline \Tstrut
15 & Twinkle Twinkle Little Star (variation 3)                                               & (20, 21) & (19, 20, 21) & $C_1$ [1/2, 1]: [0, 2, 3, 4, 19, 20, 21]                             & (20, 21) & (19, 20, 21) & $C_1$ [1/2, 1]: [0, 2, 3, 4, 19, 20, 21]                 & (20, 21) & (4, 19, 20)  & $C_1$ [1/2, 1]: [0, 1, 2, 4, 19, 20, 21]                 \Bstrut               \\
\hline \Tstrut
\multirow{3}{*}{16} & \multirow{3}{*}{\shortstack{Twinkle Twinkle Little Star \\ (variation 7)}}    & (8, 13)  & (8, 12, 13)  & $C_1$ [1/8, 3/16]: [8, 9, 10, 11, 12, 13]                            & (8, 13)  & (8, 12, 13)  & $C_1$ [1/8, 3/16]: [8, 9, 10, 11, 12, 13]                & (8, 13)  & (8, 12, 13)  & $C_1$ [1/8, 3/16]: [8, 9, 10, 11, 12, 13]                \\
   &                                                                                         & (0, 25)  & (1, 5, 25)   & $C_2$ [1/2, 3/4]: [0, 1, 2, 3, 5, 6, 7, 25]                          & (0, 25)  & (1, 5, 25)   & $C_2$ [1/2, 3/4]: [0, 1, 2, 3, 5, 6, 7, 25]              & (0, 25)  & (0, 9, 25)   & $C_2$ [1/2, 3/4]: [0, 1, 2, 3, 4, 5, 9, 25]              \\
   &                                                                                         & (23, 24) & (20, 21, 23) & $C_3$ [1/4, 3/4]: [3, 4, 5, 7, 8, 9, 20, 21, 22, 23, 24]             & (23, 24) & (20, 21, 23) & $C_3$ [1/4, 3/4]: [3, 4, 5, 7, 8, 9, 20, 21, 22, 23, 24] & (23, 24) & (9, 22, 23)  & $C_3$ [1/4, 3/4]: [3, 4, 5, 6, 7, 9, 20, 21, 22, 23, 24]     \Bstrut          \\
\hline
\end{tabular}
}
\end{center}
\caption{Comparison of the one-dimensional homology of $d_1, d_2, d_3$ for several music samples. The music titles in the table that begin with ``J-" indicate that the pieces are part of {\it Junggwangjigok}. For each song, the birth and death simplices, as well as the cycle representatives with their corresponding persistence intervals, are provided for each distance. The cycles are sorted and enumerated by their birth times in a manner such that cycles that persist only in $d_1$ (or in both $d_1 and d_3$ but die in $d_2$) are placed later in the list. }
\label{tab:compare}
\end{table}


\begin{sidewaystable}
\begin{center}
\resizebox{\textwidth}{!}{
\begin{tabular}{lclllllllll}
\hline\Tstrut\Bstrut
\multirow{2}{*}{\textbf{No}} & \multirow{2}{*}{\textbf{Song title}}& \multicolumn{3}{c} {$\mathbf{d_1} $} & \multicolumn{3}{c}{$\mathbf{d_3}$ } & \multicolumn{3}{c}{$\mathbf{d_2}$ } \\  
\cline{3-11}\Tstrut\Bstrut
 &   & \textbf{Birth}    & \textbf{Death}        & \textbf{Cycle}   & \textbf{Birth}    & \textbf{Death}        & \textbf{Cycle}    & \textbf{Birth}    & \textbf{Death}        & \textbf{Cycle}                                                        \\
\hline\Tstrut
\multirow{2}{*}{1}  & \multirow{2}{*}{\shortstack{J-Sangnyeongsan \\ (geomungo part)} }               & (4, 8)   & (2, 7, 8)    & $C_1$ [1/11, 1/8]: [2, 4, 7, 8]                                      & (4, 8)   & (2, 7, 8)    & $C_1$ [1/11, 1/8]: [2, 4, 7, 8]                          & (4, 8)   & (2, 7, 8)    & $C_1$ [1/11, 1/8]: [2, 4, 7, 8]                          \\
   &                                                                                         & (1, 11)  & (6, 8, 11)   & $C_2$ [1/5, 1/4]: [0, 1, 6, 7, 8, 11]                                &          &              &                                                               &          &              &                                                               \Bstrut               \\
\hline \Tstrut 
\multirow{4}{*}{2}  & \multirow{4}{*}{\shortstack{J-Jungnyeongsan \\ (geomungo part)}}              & (4, 10)  & (2, 6, 10)   & $C_1$ [1/7, 5/28]: [2, 4, 6, 10]                                     & (4, 10)  & (2, 6, 10)   & $C_1$ [1/7, 5/28]: [2, 4, 6, 10]                         & (4, 10)  & (4, 6, 10)   & $C_1$ [1/7, 9/56]: [2, 4, 6, 10]                         \\
   &                                                                                         & (13, 14) & (2, 13, 14)  & $C_2$ [1, 6/5]: [2, 3, 6, 8, 13, 14]                                 & (13, 14) & (2, 13, 14)  & $C_2$ [1, 6/5]: [2, 3, 6, 8, 13, 14]                     & (13, 14) & (2, 13, 14)  & $C_2$ [1, 23/20]: [2, 3, 6, 8, 13, 14]                   \\
   &                                                                                         & (15, 16) & (11, 15, 16) & $C_3$ [1, 4/3]: [2, 4, 6, 10, 11, 15, 16]                            & (15, 16) & (10, 15, 16) & $C_3$ [1, 4/3]: [10, 11, 15, 16]                         & (15, 16) & (10, 15, 16) & $C_3$ [1, 4/3]: [10, 11, 15, 16]                         \\
   &                                                                                         & (19, 20) & (3, 18, 19)  & $C_4$ [1, 2]: [3, 17, 18, 19, 20]                                    & (19, 20) & (3, 18, 19)  & $C_4$ [1, 2]: [3, 17, 18, 19, 20]                        & (19, 20) & (3, 18, 19)  & $C_4$ [1, 2]: [3, 17, 18, 19, 20]                        \Bstrut               \\
\hline \Tstrut 
\multirow{3}{*}{3}  & \multirow{3}{*}{\shortstack{J-Seryeongsan \\ (geomungo part)}}                & (3, 4)   & (0, 3, 4)    & $C_1$ [1, 14/13]: [0, 2, 3, 4]                                       & (3, 4)   & (0, 3, 4)    & $C_1$ [1, 14/13]: [0, 2, 3, 4]                           & (3, 4)   & (0, 3, 4)    & $C_1$ [1, 14/13]: [0, 2, 3, 4]                           \\
   &                                                                                         & (17, 18) & (2, 17, 18)  & $C_2$ [1, 6/5]: [0, 2, 5, 17, 18]                                    & (17, 18) & (0, 17, 18)  & $C_2$ [1, 7/6]: [0, 5, 17, 18]                           & (17, 18) & (0, 17, 18)  & $C_2$ [1, 7/6]: [0, 5, 17, 18]                           \\
   &                                                                                         & (10, 11) & (2, 9, 10)   & $C_3$ [1, 2]: [0, 2, 5, 7, 8, 9, 10, 11]                             & (10, 11) & (2, 9, 10)   & $C_3$ [1, 2]: [0, 2, 5, 7, 8, 9, 10, 11]                 & (10, 11) & (2, 9, 10)   & $C_3$ [1, 2]: [0, 2, 5, 7, 8, 9, 10, 11]                \Bstrut               \\
\hline \Tstrut  
\multirow{5}{*}{4}  & \multirow{5}{*}{\shortstack{J-Garakdeori \\ (geomungo part)}}                 & (14, 17) & (0, 7, 17)   & $C_1$ [1/2, 1]: [0, 5, 6, 7, 14, 17]                                 & (14, 17) & (0, 7, 17)   & $C_1$ [1/2, 1]: [0, 5, 6, 7, 14, 17]                     & (14, 17) & (0, 7, 17)   & $C_1$ [1/2, 1]: [0, 5, 6, 7, 14, 17]                     \\
   &                                                                                         & (5, 6)   & (2, 6, 10)   & $C_2$ [1/2, 10/9]: [0, 2, 5, 6, 10, 11, 15, 16]                      & (5, 6)   & (5, 7, 14)   & $C_2$ [1/2, 1]: [0, 2, 5, 6, 10, 11, 14, 15, 16]         & (5, 6)   & (5, 7, 14)   & $C_2$ [1/2, 1]: [0, 2, 5, 6, 10, 11, 14, 15, 16]         \\
   &                                                                                         & (16, 19) & (11, 16, 19) & $C_3$ [1, 4/3]: [0, 2, 6, 11, 14, 16, 19]                            & (16, 19) & (10, 16, 19) & $C_3$ [1, 4/3]: [2, 6, 10, 14, 16, 19]                   & (16, 19) & (10, 16, 19) & $C_3$ [1, 4/3]: [2, 6, 10, 14, 16, 19]                   \\
   &                                                                                         & (3, 4)   & (0, 3, 4)    & $C_4$ [1, 2]: [0, 2, 3, 4]                                           & (3, 4)   & (3, 4, 5)    & $C_4$ [1, 29/18]: [0, 2, 3, 4]                           & (3, 4)   & (3, 4, 16)   & $C_4$ [1, 19/12]: [0, 2, 3, 4, 16]                       \\
   &                                                                                         & (9, 10)  & (2, 8, 9)    & $C_5$ [1, 2]: [0, 6, 7, 8, 9, 10]                                    & (9, 10)  & (8, 9, 17)   & $C_5$ [1, 11/6]: [0, 6, 7, 8, 9, 10, 17]                 & (9, 10)  & (8, 9, 17)   & $C_5$ [1, 11/6]: [0, 6, 7, 8, 9, 10]                      \Bstrut               \\
\hline \Tstrut 
\multirow{7}{*}{5}  & \multirow{7}{*}{\shortstack{J-Sanghyeondodeuri \\ (geomungo part)}}           & (12, 13) & (8, 12, 13)  & $C_1$ [1/3, 2/3]: [8, 11, 12, 13]                                    & (12, 13) & (8, 11, 12)  & $C_1$ [1/3, 2/3]: [1, 8, 11, 12, 13]                     & (12, 13) & (1, 12, 13)  & $C_1$ [1/3, 8/15]: [1, 8, 11, 12, 13]                    \\
   &                                                                                         & (33, 34) & (4, 33, 34)  & $C_2$ [1, 10/9]: [4, 5, 33, 34]                                      & (33, 34) & (4, 33, 34)  & $C_2$ [1, 10/9]: [4, 5, 33, 34]                          & (33, 34) & (4, 33, 34)  & $C_2$ [1, 10/9]: [4, 5, 33, 34]                          \\
   &                                                                                         & (24, 25) & (10, 24, 25) & $C_3$ [1, 7/6]: [1, 8, 9, 10, 24, 25]                                & (24, 25) & (8, 24, 25)  & $C_3$ [1, 10/9]: [8, 9, 24, 25]                          & (24, 25) & (8, 24, 25)  & $C_3$ [1, 10/9]: [8, 9, 24, 25]                          \\
   &                                                                                         & (2, 19)  & (2, 18, 19)  & $C_4$ [1, 3/2]: [0, 1, 2, 8, 9, 10, 17, 18, 19]                      & (2, 19)  & (1, 17, 19)  & $C_4$ [1, 9/8]: [1, 2, 17, 18, 19]                       &          &              &                                                               \\
   &                                                                                         & (8, 14)  & (4, 5, 14)   & $C_5$ [1/2, 5/8]: [1, 4, 5, 8, 14]                                   &          &              &                                                               &          &              &                                                               \\
   &                                                                                         & (3, 8)   & (3, 5, 8)    & $C_6$ [1/2, 1]: [3, 4, 5, 8]                                         &          &              &                                                               &          &              &                                                               \\
   &                                                                                         & (2, 3)   & (2, 3, 8)    & $C_7$ [1, 21/20]: [1, 2, 3, 8]                                       &          &              &                                                               &          &              &                                                               \Bstrut               \\
\hline \Tstrut 
\multirow{11}{*}{6}  & \multirow{11}{*}{\shortstack{J-Hahyeondodeuri \\ (Geomungo part)}}             & (14, 26) & (7, 14, 26)  & $C_1$ [1/2, 8/15]: [1, 4, 7, 14, 26]                                 & (14, 26) & (7, 14, 26)  & $C_1$ [1/2, 8/15]: [1, 4, 7, 14, 26]                     & (14, 26) & (7, 14, 26)  & $C_1$ [1/2, 8/15]: [1, 4, 7, 14, 26]                     \\
   &                                                                                         & (15, 25) & (15, 25, 33) & $C_2$ [1/2, 3/4]: [12, 15, 25, 33]                                   & (15, 25) & (12, 15, 33) & $C_2$ [1/2, 3/4]: [12, 15, 25, 33]                       & (15, 25) & (12, 15, 33) & $C_2$ [1/2, 3/4]: [12, 15, 25, 33]                       \\
   &                                                                                         & (12, 27) & (8, 10, 27)  & $C_3$ [1, 7/6]: [4, 8, 10, 12, 27]                                   & (12, 27) & (8, 10, 27)  & $C_3$ [1, 7/6]: [4, 8, 10, 12, 27]                       & (12, 27) & (8, 10, 27)  & $C_3$ [1, 7/6]: [4, 8, 10, 12, 27]                       \\
   &                                                                                         & (6, 7)   & (6, 7, 8)    & $C_4$ [1, 4/3]: [4, 5, 6, 7, 8]                                      & (6, 7)   & (4, 6, 7)    & $C_4$ [1, 6/5]: [4, 5, 6, 7]                             & (6, 7)   & (4, 6, 7)    & $C_4$ [1, 6/5]: [4, 5, 6, 7]                             \\
   &                                                                                         & (15, 36) & (4, 33, 36)  & $C_5$ [1, 3/2]: [0, 2, 4, 15, 36]                                    & (15, 36) & (1, 15, 36)  & $C_5$ [1, 3/2]: [0, 1, 2, 15, 36]                        & (15, 36) & (0, 33, 36)  & $C_5$ [1, 83/66]: [0, 1, 15, 33, 36]                     \\
   &                                                                                         & (29, 30) & (12, 29, 30) & $C_6$ [1, 3/2]: [11, 12, 25, 29, 30]                                 & (29, 30) & (12, 29, 30) & $C_6$ [1, 3/2]: [11, 12, 25, 29, 30]                     & (29, 30) & (12, 29, 30) & $C_6$ [1, 3/2]: [11, 12, 25, 29, 30]                     \\
   &                                                                                         & (12, 19) & (4, 12, 19)  & $C_7$ [1/2, 1]: [4, 8, 12, 19]                                       & (12, 19) & (12, 19, 25) & $C_7$ [1/2, 3/4]: [4, 12, 19, 25]                        &          &              &                                                               \\
   &                                                                                         & (2, 15)  & (7, 8, 15)   & $C_8$ [1, 4/3]: [1, 2, 7, 8, 15, 25]                                 & (2, 15)  & (2, 8, 15)   & $C_8$ [1, 7/6]: [1, 2, 8, 12, 15]                        &          &              &                                                               \\
   &                                                                                         & (15, 32) & (8, 15, 32)  & $C_9$ [1, 12/11]: [8, 12, 15, 32]                                    &          &              &                                                               &          &              &                                                               \\
   &                                                                                         & (10, 11) & (8, 10, 11)  & $C_{10}$ [1, 7/6]: [4, 8, 10, 11]                                      &          &              &                                                               &          &              &                                                               \\
   &                                                                                         & (16, 17) & (8, 16, 17)  & $C_{11}$ [1, 7/6]: [4, 8, 16, 17]                                      &          &              &                                                               &          &              &                                                               \Bstrut               \\
\hline \Tstrut 
\multirow{6}{*}{7}  & \multirow{6}{*}{\shortstack{J-Yeombuldodeuri \\ (geomungo part)}}             & (5, 6)   & (4, 6, 7)    & $C_1$ [1/5, 1/3]: [2, 3, 4, 5, 6, 7]                                 & (5, 6)   & (2, 5, 7)    & $C_1$ [1/5, 7/24]: [2, 5, 6, 7]                          & (5, 6)   & (4, 5, 6)    & $C_1$ [1/5, 1/4]: [2, 3, 4, 5, 6]                        \\
   &                                                                                         & (28, 29) & (3, 28, 29)  & $C_2$ [1, 5/4]: [0, 3, 25, 28, 29]                                   & (28, 29) & (3, 28, 29)  & $C_2$ [1, 5/4]: [0, 3, 25, 28, 29]                       & (28, 29) & (3, 28, 29)  & $C_2$ [1, 5/4]: [0, 3, 25, 28, 29]                       \\
   &                                                                                         & (5, 7)   & (4, 5, 7)    & $C_3$ [1/8, 1/4]: [3, 4, 5, 7]                                       & (5, 7)   & (4, 5, 7)    & $C_3$ [1/8, 1/4]: [3, 4, 5, 7]                           &          &              &                                                               \\
   &                                                                                         & (5, 10)  & (4, 10, 11)  & $C_4$ [1/2, 7/12]: [0, 4, 5, 10, 11]                                 & (5, 10)  & (4, 10, 11)  & $C_4$ [1/2, 11/20]: [4, 5, 10, 11]                       &          &              &                                                               \\
   &                                                                                         & (3, 20)  & (6, 19, 20)  & $C_5$ [1/4, 1/2]: [2, 3, 4, 5, 6, 19, 20]                            &          &              &                                                               &          &              &                                                               \\
   &                                                                                         & (18, 19) & (4, 18, 19)  & $C_6$ [1, 26/25]: [3, 4, 18, 19]                                     &          &              &                                                               &          &              &                                                              \Bstrut               \\
\hline
\end{tabular}
}
\end{center}
\renewcommand{\thetable}{\arabic{table}a}
\caption{
Enlarged version of Table~\ref{tab:compare} (sample pieces 1--7) provided for improved readability.
}
\label{tab:x}
\end{sidewaystable}

\begin{sidewaystable}
\begin{center}
\resizebox{\textwidth}{!}{%
\begin{tabular}{lclllllllll}
\hline\Tstrut\Bstrut
\multirow{2}{*}{\textbf{No}} & \multirow{2}{*}{\textbf{Song title}}& \multicolumn{3}{c} {$\mathbf{d_1} $} & \multicolumn{3}{c}{$\mathbf{d_3}$ } & \multicolumn{3}{c}{$\mathbf{d_2}$ } \\  
\cline{3-11}\Tstrut\Bstrut
 &   & \textbf{Birth}    & \textbf{Death}        & \textbf{Cycle}   & \textbf{Birth}    & \textbf{Death}        & \textbf{Cycle}    & \textbf{Birth}    & \textbf{Death}        & \textbf{Cycle}                                                        \\
\hline\Tstrut
\multirow{6}{*}{8}  & \multirow{6}{*}{\shortstack{J-Taryeong \\ (geomungo part)}}                   & (2, 12)  & (2, 11, 12)  & $C_1$ [1/3, 10/21]: [1, 2, 10, 11, 12]                               & (2, 12)  & (2, 4, 11)   & $C_1$ [1/3, 10/21]: [1, 2, 4, 10, 11, 12]                & (2, 12)  & (2, 10, 12)  & $C_1$ [1/3, 389/1012]: [1, 2, 10, 12]                    \\
   &                                                                                         & (4, 7)   & (6, 10, 12)  & $C_2$ [1/3, 1/2]: [1, 2, 4, 6, 7, 10, 12]                            & (4, 7)   & (2, 4, 12)   & $C_2$ [1/3, 22/57]: [2, 4, 6, 7, 12]                     & (4, 7)   & (4, 10, 11)  & $C_2$ [1/3, 45/133]: [4, 7, 10, 11]                      \\
   &                                                                                         & (32, 33) & (6, 32, 33)  & $C_3$ [1, 21/20]: [1, 6, 7, 32, 33]                                  & (32, 33) & (6, 32, 33)  & $C_3$ [1, 21/20]: [1, 6, 7, 32, 33]                      & (32, 33) & (6, 32, 33)  & $C_3$ [1, 21/20]: [1, 6, 7, 32, 33]                      \\
   &                                                                                         & (27, 28) & (11, 27, 28) & $C_4$ [1, 47/35]: [0, 1, 3, 4, 11, 27, 28]                           & (27, 28) & (7, 27, 28)  & $C_4$ [1, 4/3]: [1, 3, 4, 7, 27, 28]                     & (27, 28) & (7, 27, 28)  & $C_4$ [1, 4/3]: [1, 3, 4, 7, 27, 28]                     \\
   &                                                                                         & (1, 2)   & (1, 6, 14)   & $C_5$ [1/4, 15/44]: [1, 2, 6, 14]                                    &          &              &                                                               &          &              &                                                               \\
   &                                                                                         & (22, 24) & (1, 22, 24)  & $C_6$ [1, 59/44]: [0, 1, 2, 6, 14, 22, 24]                           &          &              &                                                               &          &              &                                                              \Bstrut               \\
\hline \Tstrut 
\multirow{12}{*}{9}  & \multirow{12}{*}{\shortstack{J-Gunak \\ (geomungo part)}}                      & (0, 1)   & (1, 5, 6)    & $C_1$ [1/6, 19/90]: [0, 1, 5, 6]                                     & (0, 1)   & (1, 4, 6)    & $C_1$ [1/6, 19/90]: [0, 1, 5, 6]                         & (0, 1)   & (1, 4, 6)    & $C_1$ [1/6, 19/90]: [0, 1, 5, 6]                         \\
   &                                                                                         & (10, 11) & (4, 10, 11)  & $C_2$ [1/8, 19/88]: [3, 4, 10, 11]                                   & (10, 11) & (3, 10, 11)  & $C_2$ [1/8, 19/88]: [3, 4, 10, 11]                       & (10, 11) & (3, 10, 11)  & $C_2$ [1/8, 19/88]: [3, 4, 10, 11]                       \\
   &                                                                                         & (1, 2)   & (2, 10, 11)  & $C_3$ [1/7, 1/4]: [1, 2, 4, 5, 10, 11]                               & (1, 2)   & (2, 3, 5)    & $C_3$ [1/7, 17/70]: [1, 2, 3, 4, 5, 11]                  & (1, 2)   & (2, 3, 5)    & $C_3$ [1/7, 17/70]: [1, 2, 3, 4, 5, 11]                  \\
   &                                                                                         & (11, 16) & (3, 11, 16)  & $C_4$ [1/3, 14/33]: [0, 1, 3, 11, 16]                                & (11, 16) & (3, 11, 16)  & $C_4$ [1/3, 14/33]: [0, 1, 3, 11, 16]                    & (11, 16) & (6, 11, 16)  & $C_4$ [1/3, 2209/5544]: [0, 1, 6, 11, 16]                \\
   &                                                                                         & (16, 24) & (11, 16, 24) & $C_5$ [1/3, 2/3]: [4, 11, 16, 17, 24]                                & (16, 24) & (6, 9, 16)   & $C_5$ [1/3, 11/18]: [0, 1, 4, 5, 6, 9, 16, 17, 24]       & (16, 24) & (9, 16, 24)  & $C_5$ [1/3, 8/15]: [0, 4, 5, 9, 16, 17, 24]              \\
   &                                                                                         & (19, 31) & (11, 19, 31) & $C_6$ [1, 71/55]: [3, 7, 8, 19, 31]                                  & (19, 31) & (11, 19, 31) & $C_6$ [1, 71/55]: [3, 4, 8, 11, 19, 31]                  & (19, 31) & (8, 19, 31)  & $C_6$ [1, 139/120]: [0, 1, 3, 8, 19, 31]                 \\
   &                                                                                         & (32, 33) & (29, 32, 33) & $C_7$ [1, 23/14]: [17, 19, 29, 30, 32, 33]                           & (32, 33) & (17, 32, 33) & $C_7$ [1, 3/2]: [17, 19, 29, 30, 32, 33]                 & (32, 33) & (17, 32, 33) & $C_7$ [1, 3/2]: [17, 19, 29, 30, 32, 33]                 \\
   &                                                                                         & (25, 26) & (24, 25, 26) & $C_8$ [1, 17/10]: [0, 9, 16, 18, 24, 25, 26]                         & (25, 26) & (9, 25, 26)  & $C_8$ [1, 23/15]: [0, 1, 9, 16, 18, 25, 26]              & (25, 26) & (9, 25, 26)  & $C_8$ [1, 23/15]: [0, 9, 16, 18, 25, 26]                 \\
   &                                                                                         & (7, 19)  & (9, 11, 19)  & $C_9$ [1, 11/10]: [2, 3, 7, 9, 11, 19]                               & (6, 7)   & (6, 7, 11)   & $C_9$ [1, 12/11]: [3, 6, 7, 11]                          &          &              &                                                               \\
   &                                                                                         & (6, 7)   & (5, 6, 7)    & $C_{10}$ [1, 10/9]: [3, 5, 6, 7]                                       & (7, 19)  & (7, 11, 19)  & $C_{10}$ [1, 11/10]: [3, 4, 7, 11, 19]                     &          &              &                                                               \\
   &                                                                                         & (10, 15) & (5, 10, 15)  & $C_{11}$ [1/6, 11/56]: [4, 5, 10, 15]                                  &          &              &                                                               &          &              &                                                               \\
   &                                                                                         & (2, 8)   & (2, 3, 8)    & $C_{12}$ [1/3, 11/28]: [1, 2, 3, 8]                                    &          &              &                                                               &          &              &                                                              \Bstrut               \\
\hline \Tstrut 
\multirow{10}{*}{10} & \multirow{10}{*}{\shortstack{J-Sanghyeondodeuri \\ (haegeum part)}} & (1, 17)  & (4, 10, 17)  & $C_1$ [1/4, 1/3]: [1, 4, 10, 17]                                     & (1, 17)  & (4, 6, 17)   & $C_1$ [1/4, 23/76]: [1, 4, 6, 10, 17]                    & (1, 17)  & (4, 6, 17)   & $C_1$ [1/4, 23/76]: [1, 4, 6, 10, 17]                    \\
   &                                                                                         & (1, 11)  & (1, 11, 13)  & $C_2$ [1/5, 1/2]: [1, 4, 10, 11, 13]                                 & (1, 11)  & (4, 11, 13)  & $C_2$ [1/5, 13/42]: [1, 4, 6, 10, 11, 13]                & (1, 11)  & (4, 11, 13)  & $C_2$ [1/5, 13/42]: [1, 4, 6, 10, 11, 13]                \\
   &                                                                                         & (3, 7)   & (3, 4, 7)    & $C_3$ [1/3, 1/2]: [1, 3, 4, 7, 10, 11, 12]                           & (3, 7)   & (3, 4, 7)    & $C_3$ [1/3, 1/2]: [1, 3, 4, 6, 7, 11, 12]                & (3, 7)   & (1, 7, 12)   & $C_3$ [1/3, 1/2]: [1, 3, 4, 7, 11, 12]                   \\
   &                                                                                         & (22, 23) & (1, 22, 23)  & $C_4$ [1, 6/5]: [1, 4, 10, 21, 22, 23]                               & (22, 23) & (1, 22, 23)  & $C_4$ [1, 6/5]: [1, 4, 6, 10, 21, 22, 23]                & (22, 23) & (1, 22, 23)  & $C_4$ [1, 6/5]: [1, 6, 10, 21, 22, 23]                   \\
   &                                                                                         & (16, 27) & (2, 16, 27)  & $C_5$ [1, 5/4]: [1, 2, 16, 27]                                       & (16, 27) & (1, 16, 27)  & $C_5$ [1, 5/4]: [1, 2, 16, 27]                           & (16, 27) & (1, 16, 27)  & $C_5$ [1, 5/4]: [1, 2, 16, 27]                           \\
   &                                                                                         & (8, 9)   & (1, 8, 9)    & $C_6$ [1, 3/2]: [1, 4, 6, 7, 8, 9, 10]                               & (8, 9)   & (8, 9, 12)   & $C_6$ [1, 29/20]: [1, 6, 7, 8, 9, 12]                    & (8, 9)   & (3, 8, 9)    & $C_6$ [1, 4/3]: [1, 3, 6, 7, 8, 9]                       \\
   &                                                                                         & (4, 15)  & (6, 10, 15)  & $C_7$ [1/3, 1/2]: [4, 6, 10, 15]                                     &          &              &                                                               &          &              &                                                               \\
   &                                                                                         & (0, 15)  & (0, 10, 15)  & $C_8$ [1/2, 8/15]: [0, 1, 4, 10, 15]                                 &          &              &                                                               &          &              &                                                               \\
   &                                                                                         & (15, 25) & (6, 15, 25)  & $C_9$ [1/2, 7/12]: [6, 13, 15, 25]                                   &          &              &                                                               &          &              &                                                               \\
   &                                                                                         & (5, 6)   & (5, 6, 13)   & $C_{10}$ [1, 8/7]: [4, 5, 6, 13]                                       &          &              &                                                               &          &              &                                                               \Bstrut               \\
\hline \Tstrut 
\multirow{21}{*}{11} & \multirow{21}{*}{\shortstack{J-Hahyeondodeuri \\ (haegeum part)}}   & (20, 21) & (20, 22, 23) & $C_1$ [1/2, 7/12]: [20, 21, 22, 23]                                  & (20, 21) & (20, 22, 23) & $C_1$ [1/2, 7/12]: [20, 21, 22, 23]                      & (20, 21) & (7, 20, 21)  & $C_1$ [1/2, 7/12]: [6, 7, 20, 21]                        \\
   &                                                                                         & (18, 21) & (12, 18, 21) & $C_2$ [1/2, 3/4]: [7, 12, 18, 21]                                    & (18, 21) & (6, 18, 21)  & $C_2$ [1/2, 13/22]: [6, 7, 18, 21]                       & (18, 21) & (7, 18, 21)  & $C_2$ [1/2, 7/12]: [6, 7, 18, 21]                        \\
   &                                                                                         & (7, 20)  & (6, 20, 21)  & $C_3$ [1/2, 13/22]: [6, 7, 20, 21, 22, 23]                           & (7, 20)  & (6, 20, 21)  & $C_3$ [1/2, 13/22]: [6, 7, 20, 21, 22, 23]               & (7, 20)  & (20, 22, 23) & $C_3$ [1/2, 7/12]: [6, 7, 20, 21, 22, 23]                \\
   &                                                                                         & (17, 23) & (10, 17, 23) & $C_4$ [1/2, 3/4]: [7, 8, 10, 17, 20, 23]                             & (17, 23) & (7, 20, 23)  & $C_4$ [1/2, 3/4]: [7, 8, 9, 17, 20, 23]                  & (17, 23) & (7, 20, 23)  & $C_4$ [1/2, 3/4]: [7, 8, 9, 17, 20, 23]                  \\
   &                                                                                         & (15, 16) & (13, 15, 21) & $C_5$ [1/2, 31/30]: [4, 7, 8, 9, 10, 13, 14, 15, 16, 20, 21, 22, 23] & (15, 16) & (13, 15, 21) & $C_5$ [1/2, 31/30]: [6, 7, 8, 9, 10, 13, 14, 15, 16]     & (15, 16) & (13, 15, 21) & $C_5$ [1/2, 31/30]: [6, 7, 8, 9, 10, 13, 14, 15, 16]     \\
   &                                                                                         & (24, 26) & (4, 24, 26)  & $C_6$ [1, 3/2]: [2, 3, 4, 24, 26]                                    & (24, 26) & (4, 24, 26)  & $C_6$ [1, 3/2]: [2, 3, 4, 24, 26]                        & (24, 26) & (3, 24, 26)  & $C_6$ [1, 4/3]: [3, 4, 24, 26]                           \\
   &                                                                                         & (24, 27) & (4, 24, 27)  & $C_7$ [1, 3/2]: [2, 3, 4, 24, 27]                                    & (24, 27) & (4, 24, 27)  & $C_7$ [1, 3/2]: [2, 3, 4, 24, 27]                        & (24, 27) & (3, 24, 27)  & $C_7$ [1, 4/3]: [3, 4, 24, 27]                           \\
   &                                                                                         & (32, 33) & (21, 32, 33) & $C_8$ [1, 4/3]: [4, 7, 21, 32, 33]                                   & (32, 33) & (21, 32, 33) & $C_8$ [1, 4/3]: [4, 7, 21, 32, 33]                       & (32, 33) & (21, 32, 33) & $C_8$ [1, 4/3]: [4, 6, 7, 21, 32, 33]                    \\
   &                                                                                         & (10, 30) & (1, 9, 30)   & $C_9$ [1, 3/2]: [1, 2, 9, 10, 30]                                    & (10, 30) & (1, 9, 30)   & $C_9$ [1, 3/2]: [1, 2, 9, 10, 30]                        & (10, 30) & (1, 9, 30)   & $C_9$ [1, 3/2]: [0, 1, 2, 9, 10, 30]                     \\
   &                                                                                         & (28, 29) & (12, 28, 29) & $C_{10}$ [1, 5/3]: [0, 2, 8, 9, 12, 28, 29]                            & (28, 29) & (12, 28, 29) & $C_{10}$ [1, 5/3]: [0, 2, 7, 8, 9, 12, 28, 29]             & (28, 29) & (12, 28, 29) & $C_{10}$ [1, 268/165]: [0, 2, 7, 8, 9, 12, 28, 29]         \\
   &                                                                                         & (6, 7)   & (6, 12, 21)  & $C_{11}$ [1/2, 1]: [6, 7, 12, 21]                                      & (6, 7)   & (10, 18, 21) & $C_{11}$ [1/2, 7/10]: [6, 7, 8, 9, 10, 12, 18, 21]         &          &              &                                                               \\
   &                                                                                         & (20, 35) & (6, 20, 35)  & $C_{12}$ [1/2, 13/22]: [6, 7, 20, 35]                                  &          &              &                                                               &          &              &                                                               \\
   &                                                                                         & (14, 24) & (14, 18, 24) & $C_{13}$ [1, 31/30]: [6, 9, 10, 14, 18, 24]                            &          &              &                                                               &          &              &                                                               \\
   &                                                                                         & (5, 6)   & (5, 6, 21)   & $C_{14}$ [1, 12/11]: [4, 5, 6, 21]                                     &          &              &                                                               &          &              &                                                               \\
   &                                                                                         & (6, 19)  & (6, 19, 21)  & $C_{15}$ [1, 12/11]: [4, 6, 19, 21]                                    &          &              &                                                               &          &              &                                                               \\
   &                                                                                         & (0, 13)  & (2, 8, 20)   & $C_{16}$ [1, 8/7]: [0, 4, 8, 9, 13, 17, 20]                            &          &              &                                                               &          &              &                                                               \\
   &                                                                                         & (16, 17) & (8, 16, 17)  & $C_{17}$ [1, 8/7]: [2, 4, 8, 9, 16, 17]                                &          &              &                                                               &          &              &                                                               \\
   &                                                                                         & (22, 33) & (18, 22, 33) & $C_{18}$ [1, 8/7]: [7, 18, 22, 33]                                     &          &              &                                                               &          &              &                                                               \\
   &                                                                                         & (11, 12) & (9, 11, 12)  & $C_{19}$ [1, 6/5]: [9, 10, 11, 12]                                     &          &              &                                                               &          &              &                                                               \\
   &                                                                                         & (28, 31) & (6, 21, 28)  & $C_{20}$ [1, 4/3]: [4, 6, 21, 28, 31]                                  &          &              &                                                               &          &              &                                                               \\
   &                                                                                         & (6, 34)  & (6, 21, 34)  & $C_{21}$ [1, 4/3]: [4, 6, 21, 34]                                      &          &              &                                                               &          &              &                                                               \Bstrut               \\
\hline
\end{tabular}
}
\end{center}
  \addtocounter{table}{-1}
  \renewcommand{\thetable}{\arabic{table}b}
\caption{
Enlarged version of Table~\ref{tab:compare} (sample pieces 8--11) provided for improved readability.
}
\label{tab:y}
\end{sidewaystable}

\begin{sidewaystable}
\begin{center}
\resizebox{\textwidth}{!}{
\begin{tabular}{lclllllllll}
\hline\Tstrut\Bstrut
\multirow{2}{*}{\textbf{No}} & \multirow{2}{*}{\textbf{Song title}}& \multicolumn{3}{c} {$\mathbf{d_1} $} & \multicolumn{3}{c}{$\mathbf{d_3}$ } & \multicolumn{3}{c}{$\mathbf{d_2}$ } \\  
\cline{3-11}\Tstrut\Bstrut
 &   & \textbf{Birth}    & \textbf{Death}        & \textbf{Cycle}   & \textbf{Birth}    & \textbf{Death}        & \textbf{Cycle}    & \textbf{Birth}    & \textbf{Death}        & \textbf{Cycle}                                                        \\
\hline\Tstrut
\multirow{19}{*}{12} & \multirow{19}{*}{\shortstack{J-Yeombuldodeuri \\ (haegeum part)}}   & (0, 9)   & (8, 9, 16)   & $C_1$ [1/4, 1/2]: [0, 3, 9, 10, 14, 16]                              & (0, 9)   & (9, 10, 16)  & $C_1$ [1/4, 18/65]: [0, 9, 10, 14, 16]                   & (0, 9)   & (8, 10, 16)  & $C_1$ [1/4, 18/65]: [0, 3, 9, 10, 16]                    \\
   &                                                                                         & (5, 6)   & (2, 3, 5)    & $C_2$ [1/3, 1/2]: [0, 2, 3, 4, 5, 6]                                 & (5, 6)   & (2, 5, 18)   & $C_2$ [1/3, 13/30]: [2, 5, 6, 18]                        & (5, 6)   & (5, 6, 18)   & $C_2$ [1/3, 23/60]: [0, 5, 6, 18]                        \\
   &                                                                                         & (5, 18)  & (4, 5, 18)   & $C_3$ [1/5, 2/5]: [0, 3, 4, 5, 18]                                   & (5, 18)  & (3, 5, 18)   & $C_3$ [1/5, 2/5]: [0, 3, 4, 5, 18]                       & (5, 18)  & (0, 5, 6)    & $C_3$ [1/5, 2/5]: [0, 3, 4, 5, 18]                       \\
   &                                                                                         & (5, 10)  & (4, 10, 16)  & $C_4$ [1/3, 8/15]: [0, 4, 5, 10, 16]                                 & (5, 10)  & (3, 5, 10)   & $C_4$ [1/3, 2/5]: [3, 4, 5, 10]                          & (5, 10)  & (0, 5, 10)   & $C_4$ [1/3, 2/5]: [0, 4, 5, 10]                          \\
   &                                                                                         & (11, 13) & (10, 11, 13) & $C_5$ [1/2, 7/10]: [5, 9, 10, 11, 13]                                & (11, 13) & (11, 13, 18) & $C_5$ [1/2, 8/15]: [4, 5, 11, 13, 18]                    & (11, 13) & (11, 13, 18) & $C_5$ [1/2, 8/15]: [4, 5, 11, 13, 18]                    \\
   &                                                                                         & (19, 20) & (0, 19, 20)  & $C_6$ [1/2, 7/10]: [0, 18, 19, 20]                                   & (19, 20) & (0, 19, 20)  & $C_6$ [1/2, 7/10]: [0, 18, 19, 20]                       & (19, 20) & (18, 19, 20) & $C_6$ [1/2, 8/15]: [0, 18, 19, 20]                       \\
   &                                                                                         & (1, 12)  & (1, 9, 12)   & $C_7$ [1/2, 1]: [0, 1, 2, 9, 11, 12]                                 & (1, 12)  & (1, 9, 12)   & $C_7$ [1/2, 5/6]: [0, 1, 9, 11, 12]                      & (1, 12)  & (1, 2, 12)   & $C_7$ [1/2, 7/10]: [0, 1, 2, 9, 11, 12]                  \\
   &                                                                                         & (34, 35) & (9, 34, 35)  & $C_8$ [1/2, 5/6]: [9, 12, 34, 35]                                    & (34, 35) & (12, 34, 35) & $C_8$ [1/2, 7/10]: [0, 9, 12, 34, 35]                    & (34, 35) & (12, 34, 35) & $C_8$ [1/2, 7/10]: [0, 9, 12, 34, 35]                    \\
   &                                                                                         & (19, 39) & (5, 19, 39)  & $C_9$ [1, 4/3]: [0, 5, 8, 10, 19, 39]                                & (19, 39) & (14, 19, 39) & $C_9$ [1, 153/130]: [10, 14, 19, 39]                     & (19, 39) & (10, 19, 39) & $C_9$ [1, 267/260]: [0, 10, 19, 39]                      \\
   &                                                                                         & (37, 40) & (5, 37, 40)  & $C_{10}$ [1, 4/3]: [0, 5, 8, 10, 18, 37, 40]                           & (37, 40) & (5, 37, 40)  & $C_{10}$ [1, 4/3]: [0, 3, 5, 8, 10, 18, 37, 40]            & (37, 40) & (3, 37, 40)  & $C_{10}$ [1, 677/520]: [0, 3, 8, 10, 18, 37, 40]           \\
   &                                                                                         & (30, 31) & (0, 30, 31)  & $C_{11}$ [1, 3/2]: [0, 3, 10, 21, 24, 30, 31]                          & (30, 31) & (0, 30, 31)  & $C_{11}$ [1, 3/2]: [0, 3, 10, 21, 24, 30, 31]              & (30, 31) & (16, 30, 31) & $C_{11}$ [1, 59/40]: [0, 3, 10, 16, 21, 24, 30, 31]        \\
   &                                                                                         & (35, 41) & (0, 35, 41)  & $C_{12}$ [1, 6/5]: [0, 18, 35, 41]                                     & (35, 41) & (0, 35, 41)  & $C_{12}$ [1, 6/5]: [0, 18, 35, 41]                         &          &              &                                                               \\
   &                                                                                         & (25, 26) & (18, 25, 26) & $C_{13}$ [1, 6/5]: [0, 18, 25, 26]                                     & (25, 26) & (18, 25, 26) & $C_{13}$ [1, 6/5]: [0, 18, 25, 26]                         &          &              &                                                               \\
   &                                                                                         & (2, 14)  & (3, 14, 22)  & $C_{14}$ [1/3, 3/8]: [0, 2, 3, 14, 22]                                 &          &              &                                                               &          &              &                                                               \\
   &                                                                                         & (0, 24)  & (9, 10, 24)  & $C_{15}$ [1/2, 8/15]: [0, 9, 10, 24]                                   &          &              &                                                               &          &              &                                                               \\
   &                                                                                         & (9, 11)  & (3, 9, 11)   & $C_{16}$ [1/2, 7/10]: [2, 3, 9, 11]                                    &          &              &                                                               &          &              &                                                               \\
   &                                                                                         & (8, 28)  & (10, 28, 29) & $C_{17}$ [1, 14/13]: [8, 10, 28, 29]                                   &          &              &                                                               &          &              &                                                               \\
   &                                                                                         & (24, 40) & (5, 29, 40)  & $C_{18}$ [1, 6/5]: [5, 18, 24, 29, 40]                                 &          &              &                                                               &          &              &                                                               \\
   &                                                                                         & (24, 36) & (24, 33, 36) & $C_{19}$ [1, 6/5]: [14, 16, 24, 33, 36]                                &          &              &                                                               &          &              &                                                                \Bstrut               \\
\hline \Tstrut 
\multirow{4}{*}{13} & \multirow{4}{*}{\shortstack{Twinkle Twinkle Little Star \\ (theme)}}          & (8, 9)   & (5, 8, 9)    & $C_1$ [1/2, 5/8]: [3, 4, 5, 8, 9]                                    & (8, 9)   & (4, 8, 9)    & $C_1$ [1/2, 5/8]: [4, 5, 8, 9]                           & (8, 9)   & (4, 8, 9)    & $C_1$ [1/2, 5/8]: [4, 5, 8, 9]                           \\
   &                                                                                         & (5, 6)   & (5, 6, 7)    & $C_2$ [1/2, 3/4]: [0, 5, 6, 7]                                       & (5, 6)   & (0, 6, 7)    & $C_2$ [1/2, 3/4]: [0, 1, 4, 5, 6, 7]                     & (5, 6)   & (0, 6, 7)    & $C_2$ [1/2, 3/4]: [0, 4, 5, 6, 7]                        \\
   &                                                                                         & (11, 12) & (4, 11, 12)  & $C_3$ [1/2, 1]: [0, 1, 4, 5, 9, 10, 11, 12]                          & (11, 12) & (4, 11, 12)  & $C_3$ [1/2, 1]: [4, 5, 9, 10, 11, 12]                    & (11, 12) & (4, 11, 12)  & $C_3$ [1/2, 1]: [4, 5, 9, 10, 11, 12]                    \\
   &                                                                                         & (10, 13) & (5, 9, 13)   & $C_4$ [1/2, 1]: [5, 6, 9, 10, 13]                                    & (10, 13) & (5, 9, 13)   & $C_4$ [1/2, 1]: [5, 6, 9, 10, 13]                        & (10, 13) & (5, 9, 13)   & $C_4$ [1/2, 1]: [5, 6, 9, 10, 13]                        \Bstrut               \\
\hline \Tstrut
\multirow{6}{*}{14} & \multirow{6}{*}{\shortstack{Twinkle Twinkle Little Star \\ (variation 1)}}    & (1, 3)   & (1, 3, 14)   & $C_1$ [1/4, 5/12]: [0, 1, 3, 7, 14]                                  & (1, 3)   & (1, 3, 12)   & $C_1$ [1/4, 73/208]: [0, 1, 3, 12]                       & (1, 3)   & (1, 3, 4)    & $C_1$ [1/4, 15/52]: [0, 1, 3, 4, 12]                     \\
   &                                                                                         & (15, 16) & (11, 15, 16) & $C_2$ [1/4, 1/2]: [3, 7, 11, 12, 14, 15, 16]                         & (15, 16) & (11, 15, 16) & $C_2$ [1/4, 1/2]: [3, 4, 11, 12, 14, 15, 16]             & (15, 16) & (3, 15, 16)  & $C_2$ [1/4, 1/2]: [3, 4, 11, 12, 14, 15, 16]             \\
   &                                                                                         & (9, 12)  & (7, 12, 14)  & $C_3$ [1/6, 1/4]: [7, 9, 12, 14]                                     & (9, 12)  & (7, 9, 12)   & $C_3$ [1/6, 8/39]: [3, 7, 9, 12]                         &          &              &                                                               \\
   &                                                                                         & (7, 14)  & (8, 11, 14)  & $C_4$ [1/6, 5/16]: [3, 4, 7, 8, 12, 14]                              & (7, 14)  & (7, 12, 14)  & $C_4$ [1/6, 8/39]: [3, 7, 12, 14]                        &          &              &                                                               \\
   &                                                                                         & (11, 12) & (4, 11, 12)  & $C_5$ [1/4, 15/52]: [3, 4, 11, 12]                                   &          &              &                                                               &          &              &                                                               \\
   &                                                                                         & (13, 14) & (12, 13, 14) & $C_6$ [1/4, 5/16]: [4, 12, 13, 14]                                   &          &              &                                                               &          &              &                                                              \Bstrut               \\
\hline \Tstrut
15 & Twinkle Twinkle Little Star (variation 3)                                               & (20, 21) & (19, 20, 21) & $C_1$ [1/2, 1]: [0, 2, 3, 4, 19, 20, 21]                             & (20, 21) & (19, 20, 21) & $C_1$ [1/2, 1]: [0, 2, 3, 4, 19, 20, 21]                 & (20, 21) & (4, 19, 20)  & $C_1$ [1/2, 1]: [0, 1, 2, 4, 19, 20, 21]                 \Bstrut               \\
\hline \Tstrut
\multirow{3}{*}{16} & \multirow{3}{*}{\shortstack{Twinkle Twinkle Little Star \\ (variation 7)}}    & (8, 13)  & (8, 12, 13)  & $C_1$ [1/8, 3/16]: [8, 9, 10, 11, 12, 13]                            & (8, 13)  & (8, 12, 13)  & $C_1$ [1/8, 3/16]: [8, 9, 10, 11, 12, 13]                & (8, 13)  & (8, 12, 13)  & $C_1$ [1/8, 3/16]: [8, 9, 10, 11, 12, 13]                \\
   &                                                                                         & (0, 25)  & (1, 5, 25)   & $C_2$ [1/2, 3/4]: [0, 1, 2, 3, 5, 6, 7, 25]                          & (0, 25)  & (1, 5, 25)   & $C_2$ [1/2, 3/4]: [0, 1, 2, 3, 5, 6, 7, 25]              & (0, 25)  & (0, 9, 25)   & $C_2$ [1/2, 3/4]: [0, 1, 2, 3, 4, 5, 9, 25]              \\
   &                                                                                         & (23, 24) & (20, 21, 23) & $C_3$ [1/4, 3/4]: [3, 4, 5, 7, 8, 9, 20, 21, 22, 23, 24]             & (23, 24) & (20, 21, 23) & $C_3$ [1/4, 3/4]: [3, 4, 5, 7, 8, 9, 20, 21, 22, 23, 24] & (23, 24) & (9, 22, 23)  & $C_3$ [1/4, 3/4]: [3, 4, 5, 6, 7, 9, 20, 21, 22, 23, 24]     \Bstrut          \\
\hline
\end{tabular}
}
\end{center}
  \addtocounter{table}{-1}
  \renewcommand{\thetable}{\arabic{table}c}
\caption{
Enlarged version of Table~\ref{tab:compare} (sample pieces 12--16) provided for improved readability.
 }
\label{tab:z}
\end{sidewaystable}

}
}

\section{Conclusion}
Persistent homology deals with topological features such as loops and cycles and it provides a suitable framework for analyzing music. To apply persistent homology, a proper notion of distance is required. Although the choice of distance is important, research on how the analysis changes with different distance functions has not been thoroughly explored.

In this study, we investigated three different distance definitions that align well with our musical intuition when music data is represented as a music graph:
$d_1$, the distance defined by the minimal edge path, $d_2$,  the distance defined by the minimal total weight path, and $d_3$, an intermediate distance with hybrid characteristics of the two.

We first showed that among these, only the distance $d_2$ satisfies the formal metric definition, while the others do not. Moreover, the following inequality always holds for any pair of vertices $v, w \in V(G)$: $d_2(v, w) \le d_3(v, w) \le d_1(v, w)$. We further showed a similar inclusion relationship in the structure of one-dimensional persistence barcodes:  $\mathcal{B}_2 \subseteq \mathcal{B}_3 \subseteq \mathcal{B}_1$, where $\mathcal{B}_i$ denotes the barcode associated with the distance function $d_i$. We validated these findings using real music data by showing that the predicted relations hold exactly.

These injective relationships among commonly used distance definitions are not only theoretically interesting but also provide new insights for musical analysis. Our preliminary results suggest that analyzing music with different distance definitions in persistent homology yields different and distinct interpretations of the same music piece. Such variation could enrich both the analytical and creative processes, including automatic music generation with possibly multiple stylistic perspectives derived from the same source music.

In future work, we will further explore the musical implications of these relationships, including their roles within localized subgraphs within larger musical structures. We will also investigate automatic music generation based on varying distance functions.

\vskip .2in
\noindent
{\bf{Acknowledgements:} }
This work was supported by National Research Foundation of Korea under  grant number 2021R1A2C3009648. This work was also supported partially by the KIAS Transdisciplinary program grant. 

\vskip .2in
\noindent
{\bf{Disclosure statement:} }
The authors report there are no competing interests to declare.

\bibliographystyle{tMAM}
\bibliography{References.bib}

\end{document}